%% file: Editing-arxiv.tex
\documentclass[11pt,a4paper]{article}

\usepackage{amsmath,amsfonts,amssymb,amsthm}

\usepackage{graphicx,color}
\usepackage{boxedminipage}
\newtheorem{theorem}{Theorem}
\newtheorem{observation}{Observation}
\newtheorem{proposition}{Proposition}

\newtheorem{lemma}{Lemma}

\newcommand{\df}{\textrm{\rm def}}

\DeclareMathOperator{\operatorClassNP}{NP}
\newcommand{\classNP}{\ensuremath{\operatorClassNP}}
\DeclareMathOperator{\operatorClassFPT}{FPT}
\newcommand{\classFPT}{\ensuremath{\operatorClassFPT}}

\begin{document}

\title{Editing to a Connected Graph of Given Degrees\footnote{Supported by 
  the European Research Council (ERC) via grant Rigorous Theory of Preprocessing, reference 267959.}}

\author{Petr A. Golovach\thanks{Department of Informatics, University of Bergen, PB 7803, 5020 Bergen, Norway. E-mail: {\tt{petr.golovach@ii.uib.no}}}}

\date{}

\maketitle

\begin{abstract}
The aim of edge editing or modification problems is to change a given graph by adding and deleting of a small number of edges in order to satisfy a certain
property. We consider the \textsc{Edge Editing to a Connected Graph of Given Degrees} problem that for a given graph $G$, non-negative integers $d,k$ and a function $\delta\colon V(G)\rightarrow\{1,\ldots,d\}$, asks whether it is  possible to obtain a connected graph $G'$ from $G$ such
 that  the degree of $v$ is $\delta(v)$ for any vertex $v$ by at most $k$ edge editing operations. As the problem is NP-complete even if $\delta(v)=2$, we are interested in the parameterized complexity and show that  \textsc{Edge Editing to a Connected Graph of Given Degrees} admits a polynomial kernel when parameterized by $d+k$. For the special case $\delta(v)=d$, i.e., when the aim is to obtain a connected $d$-regular graph, the problem is shown to be fixed parameter tractable when parameterized by $k$ only. 
\end{abstract}

\section{Introduction}
The aim of graph editing or modification problems is to change a given graph as little as possible by applying specified operations in order to satisfy a certain
property.  Standard operations are vertex deletion, edge deletion, edge addition and edge contraction, but other operations are considered as well.
Various problems of this type are well-known and widely investigated. For example, such problems as  {\sc Clique}, {\sc Independent Set}, {\sc Feedback (Edge or Vertex) Set}, {\sc Cluster Editing} and many others can be seen as graph editing problems. Probably the most extensively studied variants are the problems for hereditary
properties. In particular, Lewis and Yannakakis~\cite{LewisY80} proved that for any non-trivial (in a certain sense) hereditary property,  the corresponding vertex-deletion problem is NP-hard.
The edge-deletion problems were considered by Yannakakis~\cite{Yannakakis78} and  Alon, Shapira and Sudakov~\cite{AlonSS05}. The case where edge additions and deletions are allowed and the property is the inclusion in some hereditary graph class was considered by Natanzon, Shamir and Sharan~\cite{NatanzonSS01} and Burzyn, Bonomo and Dur{\'a}n~\cite{BurzynBD06}.
The results by Cai~\cite{Cai96} and Khot and Raman~\cite{KhotR02} give a characterization of the parameterized complexity. 
For non-hereditary properties, a great deal less is known.

Moser and Thilikos in~\cite{MoserT09} and Mathieson and Szeider~\cite{MathiesonS12} initiated a study of the parameterized complexity of graph editing problems where the aim is to obtain a graph that satisfies degree constraints.  Mathieson and Szeider~\cite{MathiesonS12} considered different variants of  the following problem:
\begin{center}
\begin{boxedminipage}{.99\textwidth}
\textsc{Editing to a Graph of Given Degrees}\\
\begin{tabular}{ r l }
\textit{~~~~Instance:} & A graph $G$, non-negative integers $d,k$ and a function
                                   \\& $\delta\colon V(G)\rightarrow\{0,\ldots,d\}$.\\
\textit{Parameter 1:} & $d$.\\                                  
\textit{Parameter 2:} & $k$.\\
\textit{Question:} & Is it possible to obtain a graph $G'$ from $G$ such
 that \\& $d_{G'}(v)=\delta(v)$ for each $v\in V(G')$ by at most $k$\\& operations from the set $S$?\\
\end{tabular}
\end{boxedminipage}
\end{center}
They classified the parameterized complexity of the problem for $$S\subseteq\{\text{vertex deletion},\text{edge deletion},\text{edge addition}\}.$$ In particular,
they proved that if all the three operations are allowed, then
\textsc{Editing to a Graph of Given Degrees} is \emph{Fixed Parameter Tractable} (\classFPT) when parameterized by $d$ and $k$. Moreover, the FPT result holds for a more general version of the problem where vertices and edges have costs and the degree constraints are relaxed: for each $v\in V(G')$, $d_{G'}(v)$ should be in a given set $\delta(v)\subseteq \{1,\ldots,d\}$. 
Mathieson and Szeider also showed that \textsc{Editing to a Graph of Given Degrees} is polynomial time solvable even if $d$ and $k$ are a part of the input when only edge deletions and edge additions are allowed. 

We are interested in the following natural variant:
\begin{center}
\begin{boxedminipage}{.99\textwidth}
\textsc{Edge Editing to a Connected Graph of Given Degrees}\\
\begin{tabular}{ r l }
\textit{~~~~Instance:} & A graph $G$, non-negative integers $d,k$ and a function \\ & $\delta\colon V(G)\rightarrow\{0,\ldots,d\}$.\\
\textit{Parameter 1:} & $d$.\\                                  
\textit{Parameter 2:} & $k$.\\
\textit{Question:} & Is it possible to obtain a \emph{connected} graph $G'$ from $G$\\& such that  $d_{G'}(v)=\delta(v)$
 for each $v\in V(G')$ by at most\\& $k$ edge deletion and edge addition operations?\\
\end{tabular}
\end{boxedminipage}
\end{center}
We show that this problem is \classFPT\ when parameterized by $d$ and $k$ in Section~\ref{sec:kernel} by demonstrating a polynomial kernel of size $O(kd^3(k+d)^2)$. 
For the special case $\delta(v)=d$ for $v\in V(G)$, we call the problem  \textsc{Edge Editing to a Connected Regular Graph}.
We prove that this problem is \classFPT\ even if it is parameterized by $k$ only in Section~\ref{sec:reg}.

\section{Basic definitions and preliminaries}\label{sec:defs}

\noindent
{\bf Graphs.}
We consider only finite undirected graphs without loops or multiple
edges. The vertex set of a graph $G$ is denoted by $V(G)$ and  
the edge set  is denoted by $E(G)$.

For a set of vertices $U\subseteq V(G)$,
$G[U]$ denotes the subgraph of $G$ induced by $U$, and by $G-U$ we denote the graph obtained from $G$ by the removal of all the vertices of $U$, i.e., the subgraph of $G$ induced by $V(G)\setminus U$. 
For a non-empty set $U$, $\binom{U}{2}$ is the set of unordered pairs of distinct elements of $U$.
Also for $S\subseteq \binom{V(G)}{2}$, we say that $G[S]$ is induced by $S$, if $S$ is the set of edges of $G[S]$ and 
the vertex set of $G[S]$ is the set of vertices of $G$ incident to the pairs from $S$.
By $G-S$ we denote the graph obtained from $G$ by the removal of all the edges of $S\cap E(G)$.
Respectively, for $S\subseteq \binom{V(G)}{2}$, $G+S$ is the graph obtained from $G$ by the addition the edges that are elements of $S\setminus E(G)$.
If $S=\{a\}$, then for simplicity, we write $G-a$ or $G+a$.

For a vertex $v$, we denote by $N_G(v)$ its
\emph{(open) neighborhood}, that is, the set of vertices which are adjacent to $v$, and for a set $U\subseteq V(G)$, $N_G(U)=(\cup_{v\in U}N_G(v))\setminus U$.
The \emph{closed neighborhood} $N_G[v]=N_G(v)\cup \{v\}$, and for a positive integer $r$, $N_G^r[v]$ is the set of vertices at distance at most $r$ from $v$. 
For a set $U\subseteq V(G)$ and a positive integer $r$, $N_G^r[U]=\cup_{v\in U}N_G^r[u]$, and $N_G^r(U)=N_G^r[U]\setminus N_G^{r-1}[U]$ if $r\geq 2$.
The \emph{degree} of a vertex $v$ is denoted by $d_G(v)=|N_G(v)|$, and $\Delta(G)$ is the maximum degree of $G$. 

A  \emph{trail} in $G$ is a sequence $P=v_0,e_1,v_1,e_2,\ldots,e_s,v_s$ of vertices and edges of $G$ such that $v_0,\ldots,v_s\in V(G)$, $e_1,\ldots,e_s\in E(G)$,
the edges $e_1,\ldots,e_s$ are pairwise distinct, and 
for $i\in\{1,\ldots,s\}$, $e_i=v_{i-1}v_i$; $v_0,v_s$ are the \emph{end-vertices} of the trail. 
A trail is \emph{closed} if $v_0=v_s$. 
For $0\leq i<j\leq s$, we say that $P'=v_i,e_{i+1},\ldots,e_j,v_j$ is a \emph{segment} of $P$.
A trail is a \emph{path} if $v_0,\ldots,v_s$ are pairwise distinct except maybe $v_0,v_s$.
Sometimes we write $P=v_0,\ldots,v_s$ to denote 
 a trail $P=v_0,e_1,\ldots,e_s,v_s$ omitting edges.

A set of vertices $U$ is a \emph{cut set} of $G$ if $G-U$ has more components than $G$. A vertex $v$ is a \emph{cut vertex} if $S=\{v\}$ is a cut set. 
An edge $uv$ is a \emph{bridge} of a connected graph $G$ if $G-uv$ is disconnected.
A graph is said to be \emph{unicyclic} if it has exactly one cycle.

A set $M$ of pairwise non-adjacent edges is called a \emph{matching}, and for a bipartite graph $G$ with the given bipartition $X,Y$ of $V(G)$, a matching $M$ is \emph{perfect} (with respect to $X$) if each vertex of $X$ is incident to an edge of $M$.

\medskip
\noindent
{\bf Parameterized Complexity.}
Parameterized complexity is a two dimensional framework
for studying the computational complexity of a problem. One dimension is the input size
$n$ and the other is a parameter $k$. It is said that a problem is \emph{fixed parameter tractable} (or \classFPT), if it can be solved in time $f(k)\cdot n^{O(1)}$ for some function $f$.
A \emph{kernelization} for a parameterized problem is a polynomial algorithm that maps each instance $(x,k)$ with the input $x$ and the parameter $k$ to an instance $(x',k')$ such that i) $(x,k)$ is a YES-instance if and only if $(x',k')$ is a YES-instance of the problem, and ii) the size of $x'$ is bounded by $f(k)$ for a computable function $f$. 
The output $(x',k')$ is called a \emph{kernel}. The function $f$ is said to be a \emph{size} of a kernel. Respectively, a kernel is \emph{polynomial} if $f$ is polynomial. 
We refer to the books of 
Downey and Fellows~\cite{DowneyF13}, 
Flum and Grohe~\cite{FlumG06} and   Niedermeier~\cite{Niedermeierbook06} for  detailed introductions  to parameterized complexity.

\medskip
\noindent
{\bf Solutions of \textsc{Edge Editing to a Connected Graph of Given Degrees}.}
Let $(G,\delta,d,k)$ be an instance of \textsc{Edge Editing to a Connected Graph of Given Degrees}. 
Suppose that a connected graph $G'$ is obtained  from $G$ by at most $k$ edge deletions and edge additions such that $d_{G'}(v)=\delta(v)$ for $v\in V(G')$. 
Denote by $D$ the set of deleted edges and by $A$ the set of added edges. We say that $(D,A)$ is a \emph{solution} of \textsc{Edge Editing to a Connected Graph of Given Degrees}.
We also say that the graph $G\rq{}=G-D+A$ is obtained by editing with respect to $(D,A)$.

We need the following structural observation about solutions of \textsc{Edge Editing to a Connected Graph of Given Degrees}. 
Let $(D,A)$ be a solution for $(G,\delta,d,k)$ and let $G'=G-D+A$. We say that a trail $P=v_0,e_1,v_1,e_2,\ldots,e_s,v_s$  in $G'$ is \emph{$(D,A)$-alternating} if $e_1,\ldots,e_s\subseteq D\cup A$ 
 and for any $i\in\{2,\ldots,s\}$, either $e_{i-1}\in D,e_i\in A$ or $e_{i-1}\in A,e_i\in D$. 
We say that a $(D,A)$-alternating trail $P=v_0,e_1,v_1,e_2,\ldots,e_s,v_s$ is \emph{closed}, if $v_0=v_s$ and $s$ is even.
Notice that if $v_0=v_s$ but $s$ is odd, then such a trail is not  closed. 
Let $H(D,A)$ be the graph with the edge set $D\cup A$, and 
the vertex set of $H$ consists of vertices of $G$ incident to the edges of $D\cup A$.
Let also $Z=\{v\in V(G)|d_G(v)\neq \delta(v)\}$. 

\begin{lemma}\label{lem:alt}
For any solution $(D,A)$, the following holds.
\begin{itemize}
\item[i)] $Z\subseteq V(H(D,A))$.
\item[ii)] For any $v\in V(H(D,A))\setminus Z$, $|\{e\in D|e\text{ is incident to }v\}|=$\\ $|\{e\in A|e\text{ is incident to }v\}|$.
\item[iii)] For any $z\in Z$, $d_G(z)-\delta(z)=|\{e\in D|e\text{ is incident to }z\}|-\linebreak |\{e\in A|e\text{ is incident to }z\}|$.
\item[iv)] The graph $H(D,A)$ can be covered by a family of edge-disjoint $(D,A)$-alternating 
trails $\mathcal{T}$ (i.e., each edge of $D\cup A$ is in the unique trail of $\mathcal{T}$) and each non-closed trail in $\mathcal{T}$ has its end-vertices in $Z$. Also $\mathcal{T}$ can be constructed in polynomial time.
\end{itemize}
\end{lemma}

\begin{proof}
The claims i)--iii) are strightforward.
Because for each $v\in V(H(D,A))\setminus Z$, $|\{e\in D|e\text{ is incident to }v\}|=|\{e\in A|e\text{ is incident to }v\}|$, 
we can construct $(D,A)$-alternating trails that cover $H(D,A)$ in a greedy way. We construct trails by adding edges to each trail consecutively while it is possible. Then we start another trail until all the edges of $H(D,A)$ are covered. We choose vertices of $Z$ as starting points of trails in the beginning. If all the edges of $H(D,A)$ incident to the vertices of $Z$ are covered, then we select arbitrary vertex of $H(D,A)$ incident to uncovered edge.    
\end{proof}

\medskip
\noindent
{\bf Hardness of \textsc{Edge Editing to a Connected Graph of Given Degrees}.}
As we are interested in FPT results, we conclude this section by the observations about the classical complexity of the considered problems. 
Recall that Mathieson and Szeider proved in~\cite{MathiesonS12} that \textsc{Editing to a Graph of Given Degrees} is polynomial time solvable even if $d$ and $k$ are a part of the input when only edge deletions and edge additions are allowed. But if the obtained graph should be connected then the problem becomes \classNP-complete by an easy reduction from the {\sc Hamiltonicity} problem.

\begin{proposition}\label{prop:NPc}
For any fixed $d\geq 2$, \textsc{Edge Editing to a Connected Regular Graph} is \classNP-complete.
\end{proposition}

\begin{proof}
For simplicity, we show the claim for $d=2$. It is sufficient to observe that a graph $G$ with $n$ vertices and $m\geq n$ edges is Hamiltonian if and only if a cycle on $n$ vertices can be obtained by $m-n$ edge deletions or, equivalently, by at most $m-n$ edge deletion and edge addition operations. Since \textsc{Hamiltonicity} is a well-known \classNP-complete problem~\cite{GareyJ79}, the claim follows. Using the claim for $d=2$ as a base case, it is straightforward to show inductively that the statement holds for any $d\geq 2$.   
\end{proof}

\section{Polynomial kernel for Edge Editing to a Connected Graph of Given Degrees}\label{sec:kernel}
In this section we construct a polynomial kernel for \textsc{Edge Editing to a Connected Graph of Given Degrees} and prove the following theorem

\begin{theorem}\label{thm:kernel}
\textsc{Edge Editing to a Connected Graph of Given Degrees}  has a kernel of size $O(kd^3(k+d)^2)$.
\end{theorem}

\subsection{Technical lemmas}
We need some additional terminology. Suppose that for each $v\in V(G)$, $d_G(v)\leq \delta(v)$. If $d_G(v)<\delta(v)$, we say that $v$ is a \emph{terminal}, and $\df(v)=\delta(v)-d_G(v)$ is called a \emph{deficit} of $v$.  The \emph{deficit} $\df(G)$  of $G$ is the sum of deficits of the terminals.
We say that added edges \emph{satisfy} the deficit of a terminal $v$, if the degree of $v$ becomes equal to $\delta(v)$. 

We use the following straightforward observation.

\begin{observation}\label{obs:comp}
Let $(A,D)$ be a solution for an instance $(G,\delta,d,k)$. Then $G-D$ has at most $|A|+1$ components. 
\end{observation}

We also need the following lemma.

\begin{lemma}\label{lem:rearrange-comp}
Let $(G,\delta,d,k)$ be an instance of \textsc{Edge Editing to a Connected Graph of Given Degrees}. Suppose that there is a set $D\subseteq E(G)$ and a set $A\subseteq \binom{V(G)}{2}$ such that $G'=G-D$ has $r$ components and the following holds:
\begin{itemize}
\item[i)] for any vertex $v$ of $G''=G-D+A$, $d_{G''}(v)=\delta(v)$,
\item[ii)] $|D|+|A|\leq k$,
\item[iii)] for each component $F$ of $G'$, $\df(F)>0$,
\item[iv)] $|A|\geq r-1$. 
\end{itemize}
Then  $(G,\delta,d,k)$ has a solution.
\end{lemma}

\begin{proof}
Observe that $(D,A)$ satisfies to all conditions for solution of $(G,\delta,d,k)$ except connectivity of $G''$. Suppose that a set $A$ is chosen in such a way that the number of components of $G''$ is minimum. If $G''$ is connected, then $(D,A)$ is a solution. Assume that it is not the case, i.e., $G''$ is disconnected. Because $|A|\geq r-1$, $G''$ has a component $F$ such that there is an edge  $u_1v_1\in A$ with the property that $u_1,v_1\in V(F)$ and $u_1v_1$ is not a bridge of $F$. Let $F'$ be another component of $G''$. Because $\df(F')>0$  and for any vertex $v$ of $G''$, $d_{G''}(v)=\delta(v)$, there is and edge $u_2v_2\in A$ such that $u_2,v_2\in V(F')$. Let $A'=(A\setminus\{u_1v_1,u_2v_2\})\cup\{u_1u_2,v_1v_2\}$. Observe that $A'$ satisfies i),ii) and iv),  but $G-D+A'$ has less components that $G''$ contradicting the choice of $A$. Therefore, $G''$ is connected. 
\end{proof}

\subsection{Construction of the kernel}
Let $(G,\delta,d,k)$ be an instance of \textsc{Edge Editing to a Connected Graph of Given Degrees}. 
We assume without loss of generality that $d\geq 3$ (otherwise, we let $d=3$). 
Let 
$Z=\{v\in V(G)|d_G(v)\neq \delta(v)\}$
and $s=\sum_{v\in V(G)}|d_G(v)-\delta(v)|$.

First, we apply the following rule.

\medskip
\noindent
{\bf Rule 1.} If $|Z|>2k$ or $s$ is odd or $s>2k$ or $G$ has at least $k+2$ components, then stop and return a NO-answer.

\medskip
It is straightforward to see that the rule is safe, because each edge deletion (addition respectively) decreases (increases respectively) the degrees of two its end-vertices by one. 
Also it is clear that if $G$ has at least $k+2$ components, then at least $k+1$ edges should be added to obtain a connected graph.

From now without loss of generality we  assume that $|Z|\leq 2k$, $s\leq 2k$ 
and $G$ has at most $k+1$ components. Denote by $G_1,\ldots,G_p$ the components of $G-Z$.

Now we need some structural properties of solutions of \textsc{Edge Editing to a Connected Graph of Given Degrees}. Suppose that $(G,\delta,d,k)$ is a YES-instance, and $G'$ is a graph obtained from $G$ by the minimum number of edge deletions and edge additions such that $d_{G'}(v)=\delta(v)$ for $v\in V(G')$. Denote by $D$ the set of deleted edges and by $A$ the set of added edges. Suppose that for some $i\in \{1,\ldots,p\}$, $G_i$ has a matching $M=\{a_1,\ldots,a_s\}$ 
of size $s=\lfloor k/3\rfloor$
such that 
 $a_1,\ldots,a_s\in E(G_i-N_G^2[Z]\cap V(G_i))$  and $G_i-M$ is connected. We show that in this case we have a solution with some additional properties. 

\begin{lemma}\label{lem:sel-edges}
\textsc{Edge Editing to a Connected Graph of Given Degrees} has a solution $(D',A')$ such that 
\begin{itemize}
\item[i)] $D'\setminus E(G_i)=D\setminus E(G_i)$ and $D'\cap E(G_i)\subseteq \{a_1,\ldots,a_k\}$;
\item[ii)] for any  $uv\in A'$ such that $u\in V(G_i)$, either $u,v\in N_G(Z)$ or $v\notin V(G_i)$; 
\item[iii)] $|D'|\leq |D|$ and $|A'|\leq |A|$. 
\end{itemize}
\end{lemma}

\begin{proof}
 By Lemma~\ref{lem:alt}, $H(D,A)$ can be covered by a family of edge-disjoint $(D,A)$-alternating 
trails $\mathcal{T}$, and each non-closed trail in $\mathcal{T}$ has its end-vertices in $Z$. Denote by $\mathcal{S}$ the set of all subtrails $S=v_0,e_1,v_1,\ldots,e_s,v_s$
of trails from $\mathcal{T}$ such that $s\geq 2$, $v_0,v_s\notin V(G_i)$ and $v_1,\ldots,v_{s-1}\in V(G_i)$. Let also $\mathcal{C}$ be the set of all trails of $\mathcal{T}$ with all the vertices in $V(G_i)$. We construct $D'$ and $A'$ as follows.

First, we put in $D'$ and $A'$ all edges of $D$ and $A$ respectively that are not included in trails from $\mathcal{S}$ and $\mathcal{C}$. Then we consecutively consider trails $S\in\mathcal{S}$. Let $S=v_0,e_1,v_1,\ldots,e_s,v_s$.
We have three cases.

\medskip
\noindent
{\bf Case 1.}  $e_1,e_2\in D$. Notice that $v_0,v_s\in Z$ and $v_1,v_{s-1}\in N_G(Z)\cap V(G_i)$  in this case. Observe also that because $S$ is a $(D,A)$-alternating path, $s\geq 3$ and is odd.
We put $e_1,e_s$ in $D'$.
If $s=3$, then $e_2\in A$ and we put $e_2$ in $A'$. Clearly, $e_2=v_1v_2$ and $v_1,v_2\in N_G(Z)$. 
Suppose that $s\geq 5$.
We choose the first edge $a_j=xy\in M$ that is not included in $D'$ so far and then put $v_1x,yv_{s-1}$ in $A'$ and $a_j$ in $D'$.  Since $v_1,v_{s-1}\in N_G(Z)$ and $x,y\notin N_G^2[Z]$,
 $v_1x,yv_{s-1}\notin E(G)$. 

\medskip
\noindent
{\bf Case 2.}  $e_1\in D,e_2\in A$. In this case $v_1\in N_G(Z)\cap V(G_i)$. Also  $s\geq 2$ and is even. 
If $s=2$, then we put $e_1$ in $D'$ and $e_2$ in $A$. Let $s\geq 4$.
We choose the first edge $a_j=xy\in M$ that is not included in $D'$ so far and then put 
$e_1,a_j$ in $D'$ and $v_1x,yv_s$ in $A'$.  Since $v_1\in N_G(Z)$, $v_s\notin V(G_i)$ and $x,y\notin N_G^2[Z]$,
 $v_1x,yv_{s-1}\notin E(G)$.  

\medskip
The case $e_1\in A,e_2\in D$ is symmetric to Case 2 and is treated in the same way. It remains to consider the last case.

\medskip
\noindent
{\bf Case 3.}  $e_1,e_2\in A$.  Because $S$ is a $(D,A)$-alternating path, $s\geq 3$ and is odd.
We again choose the first edge $a_j=xy\in M$ that is not included in $D'$ so far and then put 
$a_j$ in $D'$ and $v_0x,yv_s$ in $A'$.  Since $v_0,v_s\notin V(G_i)$ and $x,y\notin N_G^2[Z]$,
 $v_1x,yv_{s-1}\notin E(G)$.  

\medskip
Observe that because $a_1,\ldots,a_s$ are pairwise non-adjacent, the edges that are included in $A'$ are distinct. Notice also that $\mathcal{S}$ has at most $k/3$ trails that have edges in $D\cap E(G_i)$. Because only for such trails $S$, we include edges of $M$ in $D\rq{}$, and for each trail, at most one edge is included, we always have edges in $M$ to include in $D\rq{}$ for the described construction.

By the construction, the sets $D',A'$ satisfy the conditions i)--iii) of the lemma. Hence, it remains to show that $(D',A')$ is a solution of the considered instance of  
\textsc{Edge Editing to a Connected Graph of Given Degrees}. By the construction, for any vertex $v\in V(G)\setminus V(G_i)$, the number of edges of $D$ incident to $v$ is the same as the number of edges of $D'$, and the number of edges of $A$ incident to $v$ is the same as the number of edges of $A'$ incident to $v$.  
Recall that $d_G(v)=\delta(v)$ for $v\in V(G_i)$. By each application of Rules 1, 2 and 3, if for a vertex $v\in V(G_i)$, we put an incident edge in $D'$, then we add an edge incident to $v$ in $A'$, and, symmetrically, if we add an edge incident to $v$ in $A'$, then we put one incident edge in $D'$. Therefore, we do not change degrees of the vertices of $G_i$ by editing. 
It follows, that if $G''=G-D\rq{}+A\rq{}$,
then for any $v\in V(G)$,
$d_{G''}(d)=d_{G'}(v)=\delta(v)$.  Finally, because the deletion of  $a_1,\ldots,a_s$ does not destroy the connectivity of $G_i$, $G''$ is connected if $G'$ is connected. 
It means that $(D',A')$ is a solution. 
\end{proof}

Using this lemma, we obtain our next rule. Let $i\in\{1,\ldots,p\}$.

\begin{figure}[ht]
\centering\scalebox{0.75}{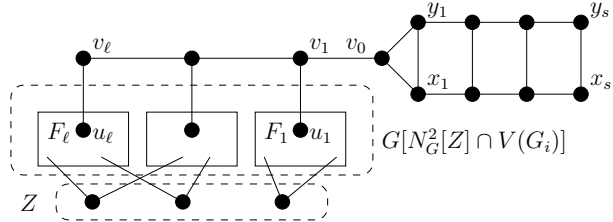}
\caption{Modification of $G_i$ by Rule 2.
\label{fig:rule-two}}
\end{figure}

\medskip
\noindent
{\bf Rule 2.}  Consider the component $G_i$ of $G-Z$.
Let $F_1,\ldots,F_{\ell}$ be the components of $G[N_G^2[Z]\cap V(G_i)]$. Notice that it can happen that $N_G^2[Z]\cap V(G_i)=\emptyset$, and it is assumed that $\ell=0$ in this case.
If $|V(G_i)|-|N_G^2[Z]\cap V(G_i)|>\ell+2k+1$, then do the following.
\begin{itemize}
\item[i)] Construct spanning trees of $F_1,\ldots,F_{\ell}$ and then construct a spanning tree $T$ of $G_i$ that contains the constructed spanning trees of $F_1,\ldots,F_{\ell}$ as subgraphs. 
\item[ii)] Let $R$ be the set of edges of $E(G_i)\setminus E(T)$ that are not incident to the vertices of $N_G^2[Z]$. Find a maximum matching $M$ in $G[R]$.
\item[iii)] If $|M|\geq k/3$, then modify $G$ and the function $\delta$ as follows (see Fig.~\ref{fig:rule-two}):
\begin{itemize}
\item delete the vertices of $V(G_i)\setminus N_G^2[Z]$;
\item construct vertices $v_0,\ldots,v_{\ell}$, $x_1,\ldots,x_s$ and $y_1,\ldots,y_s$ for $s=\lfloor k/3 \rfloor$;
\item for $j\in\{1,\ldots\ell\}$, choose a vertex $u_j$ in $F_j$ adjacent to some vertex in $G_i-V(F_i)$;
 \item construct edges $u_1v_1,\ldots,u_\ell v_\ell$,  $v_0v_1,\ldots,v_{\ell-1}v_\ell$,  $v_0x_1,v_0y_1$, $x_1x_2,\ldots,\linebreak x_{s-1}x_s$,  $y_1y_2,\ldots,y_{s-1}y_s$ and $x_1y_1,\ldots,x_sy_s$,
\item set $\delta(v_\ell)=\delta(x_s)=\delta(y_s)=2$,
$\delta(v_0)=\ldots=\delta(v_{\ell-1})=\delta(x_1)=\ldots=\delta(x_{s-1})=\delta(y_1)=\ldots=\delta(y_{s-1})=3$, $\delta(v)=d_G(v)$ (in the modified graph $G$) for $V(G_i)\cap N_G^2[Z]$, and $\delta$ has the same values as before for all other vertices of $G$.
\end{itemize}
\end{itemize}

We show that Rule 2 is safe.

\begin{lemma}\label{lem:rule-two}
Let $G'$ be the graph obtained by the application of Rule 2 from $G$ for $G_i$, and denote by $\delta'$ the modified function $\delta$.  Then
$(G',\delta',d,k)$ is a feasible instance of 
\textsc{Edge Editing to a Connected Graph of Given Degrees}, and 
$(G,\delta,d,k)$ is a YES-instance of \textsc{Edge Editing to a Connected Graph of Given Degrees} if and only if $(G',\delta',d,k)$ be a YES-instance.
\end{lemma}

\begin{proof}
Clearly, we can assume that $G$ was modified by the rule, i.e., $G_i$ was replaced by the gadget shown in Fig.~\ref{fig:rule-two}. 
Denote by $G_i'$ the component of $G'-Z$ that is obtained from $G_i$.
We have a matching $M$ of size at least $k/3$ such that $G_i-M$ is connected and edges of $M$ are not incident to vertices of $N_G^2[X]$. Let $a_1,\ldots,a_s\in M$ be $s$ arbitrary edges of $M$. 

Observe that since $d\geq 3$ and $d_{G'}(u_h)\leq d_G(u_h)$ for $h\in\{1,\ldots,\ell\}$,  $\delta'(v)\leq d$, i.e.,  $(G',\delta',d,k)$ is a feasible instance of 
\textsc{Edge Editing to a Connected Graph of Given Degrees}.

Suppose that $(G,\delta,d,k)$ is a YES-instance of \textsc{Edge Editing to a Connected Graph of Given Degrees}. Then by Lemma~\ref{lem:sel-edges},
the problem has a solution $(D,A)$ such that  $D\cap E(G_i)\subseteq \{a_1,\ldots,a_s\}$ and for any  $uv\in A$ such that $u\in V(G_i)$, either $u,v\in N_G(Z)$ or $v\notin V(G_i)$. 
Notice that each end-vertex of an edge $a_j\in D$ has the unique incident edge in $A$. Notice also that these edges of $A$ have other end-vertices outside $G_i$.
We construct the solution $(D',A')$ for $(G',\delta',d,k)$ as follows. We obtain $D'$ by replacing each edge $a_j\in D$ by the edge $x_jy_j$. 
To get $A'$, for each $a_j=f_jg_j\in D$, we replace the unique edges $hf_j,h'g_j\in A$ incident with $f_j,g_j$ by $hx_j,h'y_j$. 
As $G_i-\{x_1y_1,\ldots,x_sy_s\}$ is connected,
$N_G(Z)=N_{G'}(Z)$ and 
any two vertices $u,v\in N_G(Z)\cap V(G_i)$ are not adjacent in $G$ if and only if they are not adjacent in $G'$,  
it is straightforward to check that $(D',A')$ is a solution for $(G',\delta',d,k)$.

Assume now that  
$(G',\delta',d,k)$ is a YES-instance. 
We use the same arguments as before to construct a solution for $(G,\delta,d,k)$. Recall that 
the deletion of $x_1y_1,\ldots,x_sy_s$ does not destroy the connectivity of $G_i'$. Hence, we can apply 
Lemma~\ref{lem:sel-edges} and assume that 
$(G',\delta',d,k)$ has a solution $(D',A')$ such that  $D'\cap E(G_i')\subseteq \{x_1y_1,\ldots,x_sy_s\}$ and for any  $uv\in A'$ such that $u\in V(G_i')$, either $u,v\in N_{G'}(Z)$ or $v\notin V(G_i')$.
We construct the solution $(D,A)$ for $(G,\delta,d,k)$ as follows. We obtain $D$ by replacing each edge $x_jy_j\in D'$ by $a_j$. 
To get $A$, for each $x_jy_j\in D$, we replace the unique edges $hx_j,h'y_j\in A'$ incident with $x_j,y_j$ by $hf_j,h'g_j$ where $f_jg_j=a_j$. 
\end{proof}

We apply Rule 2 for all $i\in \{1,\ldots,p\}$. To simplify notations, assume that $(G,\delta,d,k)$ is the obtained instance of \textsc{Edge Editing to a Connected Graph of Given Degrees} and $G_1,\ldots,G_p$ are the components of $G-Z$. 
The next rule is applied to components of $G$ that are trees without vertices adjacent to $Z$. Let $i\in\{1,\ldots,k\}$.

\medskip
\noindent
{\bf Rule 3.} If $G_i$ is a tree with at least $kd/2+1$ vertices and $N_G(Z)\cap V(G_i)=\emptyset$, then replace $G_i$ by a path $P=u_1,\ldots,u_{k}$
on $k$ vertices and set $\delta(u_1)=\delta(u_{k})=1$ and $\delta(u_2)=\ldots=\delta(u_{k-1})=2$.

\begin{lemma}\label{lem:rule-three}
Let $G'$ be the graph obtained by the application of Rule 3 from $G$ for $G_i$, and denote by $\delta'$ the modified function $\delta$.  Then
$(G',\delta',d,k)$ is a feasible instance of 
\textsc{Edge Editing to a Connected Graph of Given Degrees}, and 
$(G,\delta,d,k)$ is a YES-instance of \textsc{Edge Editing to a Connected Graph of Given Degrees} if and only if $(G',\delta',d,k)$ is a YES-instance.
\end{lemma}

\begin{proof}
Obviously, we can assume that $G$ was modified by the rule, i.e., $G_i$ was replaced by a path  $P=u_1,\ldots,u_{k}$.
Since $\delta'(u_j)\leq 2$ for $j\in\{1,\ldots,k\}$, we immediately conclude that $(G',\delta',d,k)$ is a feasible instance of 
\textsc{Edge Editing to a Connected Graph of Given Degrees}.

Suppose that $(G,\delta,d,k)$ is a YES-instance of \textsc{Edge Editing to a Connected Graph of Given Degrees}, and let $(D,A)$ be a solution. Let also $\{a_1,\ldots,a_s\}=D\cap E(G_i)$. 
Notice that $s\leq k/2$ and  $\{a_1,\ldots,a_s\}\neq\emptyset$ because $G-D+A$ is a connected graph. 
Let $A_1,A_2,A_3$ be the partition of $A$ ($A_1,A_2$ can be empty) such that for any $uv\in A_1$, $u,v\in V(G_i)$, for each $uv\in A_2$, $u,v\in V(G)\setminus V(G_i)$, and the edges of $A_3$ join vertices in $V(G_i)$ with vertices in $V(G)\setminus V(G_i)$. 
Consider $F=G-D+A_1+A_2$. Notice that for any vertex $v\in V(F)$, $d_{F}(v)\leq \delta(v)$.
Denote by
$F_1$ the subgraph of $F$ induced by $V(G)\setminus V(G_i)$ and let $F_2$ the subgraph of $G-D+A_2$ induced by $V(G_i)$. By Observation~\ref{obs:comp}, $G-D$ has at most $|A_1|+|A_2|+|A_3|+1$ components. Hence, $F$ has at most $|A_3|+1$ components.
Let $t=s-|A_1|$. Because $G_i$ is a tree, $F_2$ has at least $t+1$ components.
Notice that  $\df(F_2)=2s-2|A_1|=2t$.  
Observe also that each component of $F$ has a positive deficit.

Consider the edges $a_i'=u_{2i-1}u_{2i}$ of $P$ for $i\in\{1,\ldots,t\}$. Let $F_2'$ be the graph obtained from $P$ by the deletion of the edges $a_1',\ldots,a_t'$. Notice that $F_2'$ has $t+1$ components,
$\df(F_2')=2t$, each component of $F_2'$ has a positive deficit, 
the terminals of $F_2'$, i.e., the vertices with positive deficits, are pairwise distinct and each terminal has the deficit one.

 We construct the pair $(D',A')$ where $D'\subseteq E(G')$ and $A'\subseteq \binom{V(G')}{2}\setminus E(G')$
for $(G',\delta',d,k)$ as follows. We set $D'=(D\setminus\{a_1,\ldots,a_s\})\cup\{a_1',\ldots,a_t'\}$.
Initially we include in $A'$ the edges of $A_2$.
For each terminal $u$ of $F_1$, we add edges that join this terminal with distinct terminals of $F_2'$ in the greedy way to satisfy its deficit.  Because $\df(F_2')=\df(F_2)$, we always can construct $A'$ in the described way. Moreover, $|A'|=|A|-|A_1|$ and
the number of components of $G-D$ is at least the number of components of $G'-D'$ minus $|A_1|$. Let $G''=G'-D'+A'$. For any vertex $v$ of $G''$, $d_{G''}(v)=\delta(v)$. 
By  Lemma~\ref{lem:rearrange-comp}, the instance   $(G',\delta',d,k)$  has a solution.

Suppose now that $(G',\delta',d,k)$ is a YES-instance of \textsc{Edge Editing to a Connected Graph of Given Degrees}, and let $(D',A')$ be a solution. We show that $(G,\delta,d,k)$ has a solution using symmetric arguments. 

Let $\{a_1',\ldots,a_s'\}=D\cap E(P)$. Clearly,  $s\leq k/2$ and it can be assumed that $\{a_1',\ldots,a_s'\}\neq\emptyset$. 
Observe that because $P$ is a path, we can assume without loss of generality that there is no edges in $A$ with the both end-vertices in $P$.
Because we apply the same arguments as above, it is sufficient to explain how we find the edges of $G_i$ that replace $a_1',\ldots,a_s'$ in $D'$.

The tree $G_i$ has at least $kd/2$ edges. We select $s$ pairwise non-adjacent edges $a_1,\ldots,a_s$ in $G_i$ in the greedy way: we select an edge incident with a leaf and then delete the edge, the incident vertices, and the adjacent edges. As the maximum degree of $G_i$ is at most $d$ and $s\leq k/2$, we always find $a_1,\ldots,a_s$.  
Let $F_2'$ be the graph obtained from $G_i$ by the deletion of the edges $a_1,\ldots,a_s$. Notice that $F_2'$ has $s+1$ components,
$\df(F_2')=2s$, each component of $F_2'$ has a positive deficit, 
the terminals of $F_2'$, i.e., the vertices with positive deficits, are pairwise distinct and each terminal has the deficit one.
\end{proof}

The next rule is applied to components of $G$ that are unicyclic graphs without vertices adjacent to $Z$. Let $i\in\{1,\ldots,k\}$.

\medskip
\noindent
{\bf Rule 4.} If $G_i$ is a unicyclic graph with at least $kd/2$ vertices and $N_G(Z)\cap V(G_i)=\emptyset$, then replace $G_i$ by a cycle $C=u_0,\ldots,u_{k}$
on $k$ vertices, $u_0=u_k$, and set $\delta(u_1)=\ldots=\delta(u_{k})=2$.

\begin{lemma}\label{lem:rule-four}
Let $G'$ be the graph obtained by the application of Rule 4 from $G$ for $G_i$, and denote by $\delta'$ the modified function $\delta$.  Then
$(G',\delta',d,k)$ is a feasible instance of 
\textsc{Edge Editing to a Connected Graph of Given Degrees}, and 
$(G,\delta,d,k)$ is a YES-instance of \textsc{Edge Editing to a Connected Graph of Given Degrees} if and only if $(G',\delta',d,k)$ is a YES-instance.
\end{lemma}

\begin{proof}
Obviously, we can assume that $G$ was modified by the rule, i.e., $G_i$ was replaced by a cycle  $C=u_0,\ldots,u_{k}$.
Since $\delta'(u_j)\leq 2$ for $j\in\{1,\ldots,k\}$, we immediately conclude that $(G',\delta',d,k)$ is a feasible instance of 
\textsc{Edge Editing to a Connected Graph of Given Degrees}.

Suppose that $(G,\delta,d,k)$ is a YES-instance of \textsc{Edge Editing to a Connected Graph of Given Degrees}, and let $(D,A)$ be a solution. Let also $\{a_1,\ldots,a_s\}=D\cap E(G_i)$. Notice that $s\leq k/2$ and  $\{a_1,\ldots,a_s\}\neq\emptyset$. 
Let $A_1,A_2,A_3$ be the partition of $A$ ($A_1,A_2$ can be empty) such that for any $uv\in A_1$, $u,v\in V(G_i)$, for each $uv\in A_2$, $u,v\in V(G)\setminus V(G_i)$, and the edges of $A_3$ join vertices in $V(G_i)$ with vertices in $V(G)\setminus V(G_i)$. 
Consider $F=G-D+A_1+A_2$. Notice that for any vertex $v\in V(F)$, $d_{F}(v)\leq \delta(v)$.
Denote by
$F_1$ the subgraph of $F$ induced by $V(G)\setminus V(G_i)$ and let $F_2$ the subgraph of $G-D+A_2$ induced by $V(G_i)$. By Observation~\ref{obs:comp}, $G-D$ has at most $|A_1|+|A_2|+|A_3|+1$ components. Hence, $F$ has at most $|A_3|+1$ components.
Let $t=s-|A_1|$. Because $G_i$ is a unicyclic graph, $F_2$ has at least $t$ components.
Notice that  $\df(F_2)=2s-2|A_1|=2t$.  
Observe also that each component of $F$ has a positive deficit.

Consider the edges $a_i'=u_{2i-1}u_{2i}$ of $C$ for $i\in\{1,\ldots,t\}$. Let $F_2'$ be the graph obtained from $C$ by the deletion of the edges $a_1',\ldots,a_t'$. Notice that $F_2'$ has $t$ components,
$\df(F_2')=2t$, each component of $F_2'$ has a positive deficit, 
the terminals of $F_2'$, i.e., the vertices with positive deficits, are pairwise distinct and each terminal has the deficit one.

We construct the pair $(D',A')$ where $D'\subseteq E(G')$ and $A'\subseteq \binom{V(G')}{2}\setminus E(G')$
for $(G',\delta',d,k)$ as follows. We set $D'=(D\setminus\{a_1,\ldots,a_s\})\cup\{a_1',\ldots,a_t'\}$.
Initially we include in $A'$ the edges of $A_2$.
For each terminal $u$ of $F_1$, we add edges that join this terminal with distinct terminals of $F_2'$ in the greedy way to satisfy its deficit.  Because $\df(F_2')=\df(F_2)$, we always can construct $A'$ in the described way. Moreover, $|A'|=|A|-|A_1|$ and
the number of components of $G-D$ is at least the number of components of $G'-D'$ minus $|A_1|$. Let $G''=G'-D'+A'$. For any vertex $v$ of $G''$, $d_{G''}(v)=\delta(v)$. 
By  Lemma~\ref{lem:rearrange-comp}, the instance   $(G',\delta',d,k)$  has a solution.

Suppose now that $(G',\delta',d,k)$ is a YES-instance of \textsc{Edge Editing to a Connected Graph of Given Degrees}, and let $(D',A')$ be a solution. We show that $(G,\delta,d,k)$ has a solution using symmetric arguments. 

Let $\{a_1',\ldots,a_s'\}=D\cap E(C)$. Clearly,  $s\leq k/2$ and it can be assumed that $\{a_1',\ldots,a_s'\}\neq\emptyset$. 
Observe that because $C$ is a cycle, we can assume without loss of generality that there is no edges in $A$ with the both end-vertices in $C$.
Because we apply the same arguments as above, it is sufficient to explain how we find the edges of $G_i$ that replace $a_1',\ldots,a_s'$ in $D'$.

The graph $G_i$ has at least $kd/2$ edges. We select $s$ pairwise non-adjacent edges $a_1,\ldots,a_s$ in $G_i$ in the greedy way. First, we select an edge $a_1$ in the unique cycle of $G_i$ and delete it together with incident vertices and adjacent edges. Observe that we delete at most $2d-1$ edges. Then  
 we recursively select an edge in the obtained incident with a leaf and then delete the edge, the incident vertices, and the adjacent edges. As the maximum degree of $G_i$ is at most $d$ and $s\leq k/2$, we always find $a_1,\ldots,a_s$.   
Let $F_2'$ be the graph obtained from $G_i$ by the deletion of the edges $a_1,\ldots,a_s$. Notice that $F_2'$ has $s$ components,
$\df(F_2')=2s$, each component of $F_2'$ has a positive deficit, 
the terminals of $F_2'$, i.e., the vertices with positive deficits, are pairwise distinct and each terminal has the deficit one.
\end{proof}

The Rules 3 and 4 are applied for all $i\in\{1,\ldots,p\}$. We again assume that $(G,\delta,d,k)$ is the obtained instance of \textsc{Edge Editing to a Connected Graph of Given Degrees} and $G_1,\ldots,G_p$ are the components of $G-Z$.

We construct the set of \emph{branch} vertices $B=B_{1}\cup B_2$. Let $\hat{G}$ be the graph obtained from $G$ by the recursive deletion of vertices $V(G)\setminus N_G[Z]$ of degree one or zero. A vertex $v\in V(\hat{G})$ is included in $B_1$ if  $d_{\hat{G}}(v)\geq 3$ or $v\in N_{\hat{G}}[Z]$,
and $v$ is included in $B_2$ if $v\notin B_1$, $d_{\hat{G}}(v)=2$ and there are $x,y\in B_1$ (possibly $x=y$) such that $v$ is in a $(x,y)$-path of length at most 6.

We apply the following rules 5 and 6 for each $v\in B$.

\begin{figure}[ht]
\centering\scalebox{0.7}{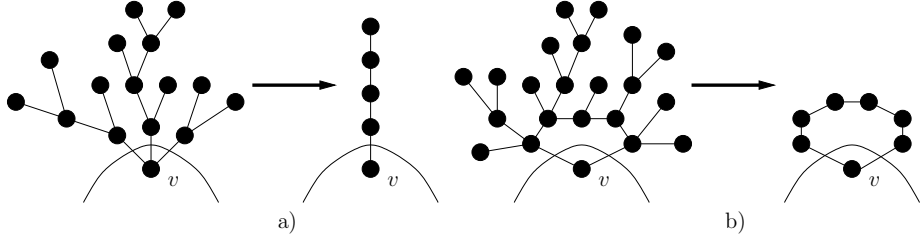}
\caption{Modification of $G$ by Rule 5 and Rule 6.
\label{fig:rule-five-six}}
\end{figure}

\medskip
\noindent
{\bf Rule 5.} If $v\in B$ is a cut vertex of $G$, then find all components of $T_1,\ldots,T_\ell$ of $G-v$ such that for $i\in\{1,\ldots,\ell\}$, i) $T_i$ is a tree, ii) $V(T_i)\subseteq V(G)\setminus B$, and iii) $T_i$ has the unique vertex $v_i$ adjacent to $v$. Let $T=G[V(T_1)\cup\ldots\cup V(T_\ell)\cup\{v\}]$ (see Fig.~\ref{fig:rule-five-six} (a). If the tree $T$ has at least $kd/2+d^2$ vertices, then 
replace $T_1,\ldots,T_\ell$ by a path $P=u_0,u_1,\ldots,u_{k}$, join $v$ and $u_0$ by an edge, and set $\delta(u_0)=\delta(u_{k-1})=2$, $\delta(u_k)=1$, and $\delta(v)=d_G(v)-\ell+1$.

\begin{lemma}\label{lem:rule-five}
Let $G'$ be the graph obtained by the application of Rule 5 from $G$ for $v\in B$, and denote by $\delta'$ the modified function $\delta$.  Then
$(G',\delta',d,k)$ is a feasible instance of 
\textsc{Edge Editing to a Connected Graph of Given Degrees}, and 
$(G,\delta,d,k)$ is a YES-instance of \textsc{Edge Editing to a Connected Graph of Given Degrees} if and only if $(G',\delta',d,k)$ is a YES-instance.
\end{lemma}

\begin{proof}
Clearly, we can assume that $G$ was modified by the rule, i.e., $T_1,\ldots,T_\ell$ were replaced by a path  $P=u_0,\ldots,u_{k}$.
Since $\delta'(u_j)\leq 2$ for $j\in\{1,\ldots,k+1\}$ and $d_{G'}(v)=d_G(v)=\delta(v)$, we have that $(G',\delta',d,k)$ is a feasible instance of 
\textsc{Edge Editing to a Connected Graph of Given Degrees}. Denote by $T$ the tree $G[V(T_1)\cup\ldots\cup V(T_\ell)\cup \{v\}]$.

Suppose that $(G,\delta,d,k)$ is a YES-instance of \textsc{Edge Editing to a Connected Graph of Given Degrees}, and let $(D,A)$ be a solution. 
If $D\cap E(T)=\emptyset$, then  $(D,A)$ is a solution for $(G',\delta',d,k)$. Hence, let 
$D\cap E(T)=\{a_1,\ldots,a_s\}\neq\emptyset$. 
Denote by $h$ the number of edges of $D\cap E(T)$ incident to $v$. 
The vertex $v$ has at least $h$ edges of $A$ incident to $v$. Let $B$ be the set obtained from $A$ by the deletion of $h$ such edges and assume that $h_1$ selected edges join $v$ with  vertices of $T-v$ and the remaining $h_2=h-h_1$ edges join $v$ with vertices of $G-V(G_i)$. 
Let $A_1,A_2,A_3$ be the partition of $B$ ($A_1,A_2$ can be empty) such that for any $xy\in A_1$, $x,y\in V(T)\setminus\{v\}$, for each $xy\in A_2$, $x,y\in (V(G)\setminus V(T))\cup\{v\}$,
 and the edges of $A_3$ join vertices in $V(T)\setminus\{v\}$ with vertices in $(V(G)\setminus V(T))\cup\{v\}$.

Consider $F=G-D+A_1+A_2$. Notice that for any vertex $x\in V(F)$, $d_{F}(x)\leq \delta(x)$.
By Observation~\ref{obs:comp}, $G-D$ has at most $|A_1|+|A_2|+|A_3|+h+1$ components. Hence, $F$ has at most $|A_3|+h+1$ components. 

Denote by $F_1$ the subgraph of $F$ induced by $V(T)$. 
Let $t=s-h_1-|A_1|$. Because $T$ is a tree, $F_1$ has at least $s+1-|A_1|$ components.
Notice that the total deficit of the terminals in $V(T)\setminus\{v\}$ is $2s-h_1-2|A_1|=2t+h_1$. 
Notice that because $(D,A)$ is a solution, the remaining terminals have the same total deficit and it is equal $|A_3|+h_2$.
Observe also that each component of $F$ has a positive deficit.

Consider the edges $a_i'=u_{2i-1}u_{2i}$ of $P$ for $i\in\{1,\ldots,t\}$. Let $F_2$ be the graph obtained from $P$ by the deletion of the edges $a_1',\ldots,a_s'$. Notice that $F_2$ has $t$ components,
$\df(F_2)=2t$, each component of $F_2$ has a positive deficit, 
the terminals of $F_2$, i.e., the vertices with positive deficits, are pairwise distinct and each terminal has the deficit one.

We construct the pair $(D',A')$ where $D'\subseteq E(G')$ and $A'\subseteq \binom{V(G')}{2}\setminus E(G')$
for $(G',\delta',d,k)$ as follows. We set $D'=(D\setminus\{a_1,\ldots,a_s\})\cup\{a_1',\ldots,a_t'\}$.
Initially we include in $A'$ the edges of $A_2$.
For each terminal $u$ of $G-V(G_i)$, we add edges that join this terminal with distinct terminals of $F_2'$ in the greedy way to satisfy its deficit. Notice that we add $|A_3|+h_2$ edges.
By the choice of $t$,  we always can construct $A'$ in the described way. Moreover, $|A'|=|A|-h_1-|A_1|+h_2$ and
the number of components of $G-D$ is at least the number of components of $G'-D'$ minus $h_1+|A_1|$. Let $G''=G'-D'+A'$. For any vertex $v$ of $G''$, $d_{G''}(v)=\delta(v)$. 
By  Lemma~\ref{lem:rearrange-comp}, the instance   $(G',\delta',d,k)$  has a solution.

Suppose now that $(G',\delta',d,k)$ is a YES-instance of \textsc{Edge Editing to a Connected Graph of Given Degrees}, and let $(D',A')$ be a solution. We show that $(G,\delta,d,k)$ has a solution using symmetric arguments. 

Let $\{a_1',\ldots,a_s'\}=D\cap (E(P)\cup\{vu_0\})$. Clearly,  $s\leq k/2$ and it can be assumed that $\{a_1',\ldots,a_s'\}\neq\emptyset$. 
Because $P$ is a path, it can be assumed that $a_i'=u_{2i-1}u_{2i}$ for $i\in\{1,\ldots,s\}$, and it simplifies the arguments.
Observe also that because $P$ is a path, we can assume without loss of generality that there is no edges in $A$ with the both end-vertices in $P$ except, possibly, $u_1u_{2s}$. In the last case $t=s-1$ and otherwise $t=s$.
Because we apply the same arguments as above, it is sufficient to explain how we find the edges of $G_i$ that replace $a_1',\ldots,a_s'$ in $D'$.
The trees $T_1,\ldots,T_\ell$ have at least $kd/2+d^2$ edges. We select $t$ pairwise non-adjacent edges $a_1,\ldots,a_t$ in $T$ that are not incident to $v_1,\ldots,v_\ell$
 in the greedy way: we select an edge incident with a leaf that is not adjacent to $v_1,\ldots,v_\ell$ and then delete the edge, the incident vertices, and the adjacent edges. As the maximum degree of $T$ is at most $d$ and $s\leq k/2$, we always find $a_1,\ldots,a_t$.   
\end{proof}

\medskip
\noindent
{\bf Rule 6.} If $v\in B$ is a cut vertex of $G$ and there is a component $T$ of $G-v$ such that  i) $T$ is a tree, ii) $V(T)\subseteq V(G)\setminus B$,  iii) $T$ has exactly two vertices adjacent to $v$,
and iv)  $|V(T)|\geq (k/2+2)d+1$, then 
replace $T$ by a path $P=u_0,\ldots,u_{k+1}$,  join $v$ and $u_0,u_{k+1}$ by edges, and 
set $\delta(u_0)=\ldots=\delta(u_{k+1})=2$ (see Fig.~\ref{fig:rule-five-six} (b).

\medskip
The next rule is applied to pairs of distinct vertices $u,v\in B$.

\begin{figure}[ht]
\centering\scalebox{0.7}{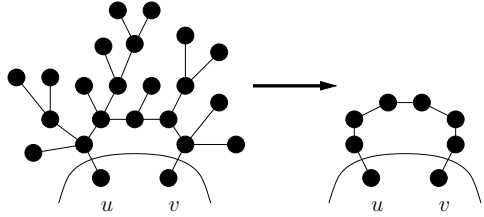}
\caption{Modification of $G$ by Rule 7.
\label{fig:rule-seven}}
\end{figure}

\medskip
\noindent
{\bf Rule 7.} If $\{u,v\}\in B$ is a cut set of $G$ and there is a component $T$ of $G-\{u,v\}$ such that  i) $T$ is a tree, ii) $V(T)\subseteq V(G)\setminus B$,  iii) $T$ has a
unique vertex adjacent to $u$ and a unique vertex  adjacent to $v$ and has no vertices adjacent to  both $u$ and $v$,
and iv)  $|V(T)|\geq (k/2+2)d+1$, then 
replace $T$ by a path $P=u_0,\ldots,u_{k+1}$,  join $u$ with  $u_0$ and $v$ with $u_{k+1}$ by edges, and 
set $\delta(u_0)=\ldots=\delta(u_{k+1})=2$ (see Fig.~\ref{fig:rule-seven}).

\begin{lemma}\label{lem:rule-six}
Let $G'$ be the graph obtained by the application of Rule 6 from $G$ for $v\in B$ (Rule 7 from $G$ for $u,v\in B$), and denote by $\delta'$ the modified function $\delta$.  Then
$(G',\delta',d,k)$ is a feasible instance of 
\textsc{Edge Editing to a Connected Graph of Given Degrees}, and 
$(G,\delta,d,k)$ is a YES-instance of \textsc{Edge Editing to a Connected Graph of Given Degrees} if and only if $(G',\delta',d,k)$ is a YES-instance.
\end{lemma}

\begin{proof}
We prove the lemma for Rule 7. The proof for Rule 6 is done by the same arguments with the assumption that $u=v$.
We can assume that $G$ was modified by the rule, i.e., a tree $T$ was replaced by a path  $P=u_0,\ldots,u_{k+1}$. 
Since $\delta'(u_j)\leq 2$ for $j\in\{0,\ldots,k+1\}$, $d_{G'}(u)=d_G(u)=\delta(u)$ and $d_{G'}(v)=d_G(v)=\delta(v)$ we have that $(G',\delta',d,k)$ is a feasible instance of 
\textsc{Edge Editing to a Connected Graph of Given Degrees}.

Let $u'',v''$ be the vertices of $T$ adjacent to $u$ and $v$ respectively, and denote by $T'$ the graph obtained from $T$ by the addition the vertices $u,v$ and the edges $uu',vv'$.
Denote by $R$ the unique $(u,v)$-path that goes through $T'$. By the definition of $B$, $R$ has length at least 7.

Suppose that $(G,\delta,d,k)$ is a YES-instance of \textsc{Edge Editing to a Connected Graph of Given Degrees}. 
Let $(D,A)$ be a solution. If $D$ has no edges incident to the vertices of $T$, then $(D,A)$ is a solution for $(G',\delta',d,k)$. Hence, we can assume that  
$D\cap E(T')=\{a_1,\ldots,a_s\}\neq\emptyset$. 

If $uu'\in D$, then we select an edge of $A$ incident to $u$, and if $vv'\in D$, then we select an edge of $A$ incident to  $v$ (notice that if $uu',vv'\in D$, we the same edge could be selected). 
Denote by $B$ the set obtained from $A$ by the deletion of the selected edges, and let 
$h$ be the number of selected edges and assume that $h_1$ selected edges join $u$ and $v$ with  vertices of $T$ and the remaining $h_2=h-h_1$ edges join $u,v$ with vertices of $G-V(G_i)$. 
Let $A_1,A_2,A_3$ be the partition of $B$ ($A_1,A_2$ can be empty) such that for any $xy\in A_1$, $x,y\in V(T)$, for each $xy\in A_2$, $x,y\in (V(G)\setminus V(T))$,
 and the edges of $A_3$ join vertices in $V(T)$ with vertices in $(V(G)\setminus V(T))$.

Consider $F=G-D+A_1+A_2$. Notice that for any vertex $x\in V(F)$, $d_{F}(x)\leq \delta(x)$.
By Observation~\ref{obs:comp}, $G-D$ has at most $|A_1|+|A_2|+|A_3|+h+1$ components. Hence, $F$ has at most $|A_3|+h+1$ components. 

Denote by $F_1$ the subgraph of $F$ obtained from $T'$ by the deletion of edges of $D$. 
Let $t=s-h_1-|A_1|$. Because $T$ is a tree, $F_1$ has at least $s+1-|A_1|$ components if $u\neq v$ and at least $s-|A_1|$ components if $u=v$.
Notice that the total deficit of the terminals in $V(T)$ is $2s-h_1-2|A_1|$. Notice that because $(D,A)$ is a solution, the remaining terminals have the same total deficit and it is equal $|A_3|+h_2$.
Observe also that each component of $F$ has a positive deficit.

Consider the edges $a_i'=u_{2i-1}u_{2i}$ of $P$ for $i\in\{1,\ldots,t\}$. Let $F_2$ be the graph obtained from $P$ by the deletion of the edges $a_1',\ldots,a_t'$. Notice that $F_2$ has $t$ components,
$\df(F_2)=2t$, each component of $F_2$ has a positive deficit, 
the terminals of $F_2$, i.e., the vertices with positive deficits, are pairwise distinct and each terminal has the deficit one.

We construct the pair $(D',A')$ where $D'\subseteq E(G')$ and $A'\subseteq \binom{V(G')}{2}\setminus E(G')$
for $(G',\delta',d,k)$ as follows. We set $D'=(D\setminus\{a_1,\ldots,a_t\})\cup\{a_1',\ldots,a_s'\}$.
Initially we include in $A'$ the edges of $A_2$.
For each terminal $u$ of $G-V(G_i)$, we add edges that join this terminal with distinct terminals of $F_2'$ in the greedy way to satisfy its deficit. Notice that we add $|A_3|+h_2$ edges.
By the choice of $t$,  we always can construct $A'$ in the described way. Moreover, $|A'|=|A|-h_1-|A_1|+h_2$ and
the number of components of $G-D$ is at least the number of components of $G'-D'$ minus $h_1+|A_1|$. Let $G''=G'-D'+A'$. For any vertex $v$ of $G''$, $d_{G''}(v)=\delta(v)$. 
By  Lemma~\ref{lem:rearrange-comp}, the instance   $(G',\delta',d,k)$  has a solution.

Suppose now that $(G',\delta',d,k)$ is a YES-instance of \textsc{Edge Editing to a Connected Graph of Given Degrees}, and let $(D',A')$ be a solution. We show that $(G,\delta,d,k)$ has a solution using symmetric arguments. 

Let $\{a_1',\ldots,a_s'\}=D\cap (E(P)\cup\{vu_0\})$. Clearly,  $s\leq k/2$ and it can be assumed that $\{a_1',\ldots,a_s'\}\neq\emptyset$. 
Because $P$ is a path, it can be assumed that $a_i'=u_{2i-1}u_{2i}$ for $i\in\{1,\ldots,s\}$, and it simplifies the arguments.
Observe also that because $P$ is a path, we can assume without loss of generality that there is no edges in $A$ with the both end-vertices in $P$ except, possibly, $u_1u_{2s}$. In the last case $t=s-1$ and otherwise $t=s$.
Because we apply the same arguments as above, it is sufficient to explain how we find the edges of $G_i$ that replace $a_1',\ldots,a_s'$ in $D'$.

The tree $T$ has at least $kd/2+2d$ edges. We select $s$ pairwise non-adjacent edges $a_1,\ldots,a_s$ in $T$ that are not incident to $u',v'$
in the greedy way. Let $u''\neq u$ be the vertex incident to $u'$ in $R$ and let $v''\neq v$ be the vertex incident to $v'$ in $R$. Let $e_1\neq u'u''$ be the edge incident to $u''$ in $R$ and 
let $e_2\neq v'v''$ be the edge incident to $v''$ in $R$. Because $R$ has length at least 7, $e_1\neq e_2$ and these edges are not adjacent.
We start the greedy choice by selecting $a_1=e_1$ and $a_2=e_2$, then we delete them together with the incident vertices and the adjacent edges. We proceed by 
 selecting an edge incident with a leaf that is not adjacent to the vertices adjacent to $u,v$, 
 and then delete the edge, the incident vertices, and the adjacent edges. 
As the maximum degree of $T$ is at most $d$ and $s\leq k/2$, we always find $a_1,\ldots,a_s$.   
\end{proof}

It is straightforward to see that Rules 1--7 can be applied in polynomial time. Also Lemmas~\ref{lem:rule-two}--\ref{lem:rule-six} prove that we obtain an equivalent instance of  \textsc{Edge Editing to a Connected Graph of Given Degrees}. To show that we have a polynomial kernel, it remains to get an upper bound for the size of the obtained graph.

\begin{lemma}\label{lem:size}
Let $(G\rq{},\delta\rq{},d,k)$ be the instance of \textsc{Edge Editing to a Connected Graph of Given Degrees} obtained
from $(G,\delta,d,k)$. Then $|V(G\rq{})|=O(kd^3(k+d)^2)$.
\end{lemma}

\begin{proof}
By Rule~1, we have that $|Z|\leq 2k$ and $|N(Z)|\leq 2kd+s\leq 2k(d+1)$. Respectively, $|N^2_G[Z]|\leq 2k(d(d+1)+1)$ and $|N_G^3(Z)|\leq 2k(d+1)(d-1)^2$. Also by Rule~1,
the number of components of $G-Z$ is at most $|N_G(Z)|+k\leq 2k(d+1)+k$.

Consider a component $G_i$ of $G-Z$, $i\in\{1,\ldots,p\}$. Let $G_i\rq{}$ be the graph obtained from $G_i$ by the deletion of the vertices of $N_G^2[Z]$. Notice that if $|E(G_i\rq{})|-|V(G_i\rq{})|\geq (2d-3)k/3-1$, then for any spanning tree $T$ of $G_i$, there is a matching $M\subseteq E(G_i)\setminus E(T)$ of size at least $k/3$, because $M$ can be constructed by the greedy algorithm. 
Since the number of components of $G[V(G_i)\cap N_G^2[Z]]$ is at most $|N_G(Z)|\leq 2k(d+1)$, if additionally $|V(G_i\rq{})|>2k(d+2)+1$, then $G_i$ is modified by Rule~2.

Notice that we have the worst case if Rules 3 and 4 are not applied.
We have the worst case if we do not apply this rule and Rules~3 and 4.
It follows that in the worst case $G$ has at most $b_1=2(|N_G^3(Z)|+2((2d-3)k-2)(2k(d+1)+k))-2$ branch vertices in $B_1\cap (V(G)\setminus N_G^2[Z])$. Also in the worst case Rules~6 and 7 are not applied, and all other vertices of $G_1\rq{},\ldots,G_p\rq{}$ that are on the paths, that join vertices of $B_1$ with each other, are in $B_2$, and we have at most $b_2=4(b_1-1)$ such vertices.  Taking into account Rule~5, we have that $|V(G_1\rq{})|+\ldots+|V(G_p\rq{})|\leq (b_1+b_2)(kd/2+d^2+1)$. 
 We conclude that $|V(G\rq{})|=O(kd^3(k+d)^2)$. 
\end{proof}

\section{FPT algorithm for Edge Editing to a Connected Regular Graph}\label{sec:reg}
In this section we construct an \classFPT-algorithm for {\sc Edge Editing to a Connected Regular Graph} with the parameter $k$ ($d$ is a part of the input here) and prove the following theorem.

\begin{theorem}\label{thm:regular}
 {\sc Edge Editing to a Connected Regular Graph} can be solved in time $O^*(k^{O(k^3)})$.
\end{theorem}

\subsection{Preliminaries}\label{sec:prelim-alg}
We need the result obtained by Mathieson and Szeider in~\cite{MathiesonS12}.  Let $G$ be a graph, and let $\rho\colon \binom{V(G)}{2}\rightarrow\mathbb{N}$ be a \emph{cost} function that for any two distinct vertices $u,v$ defines the cost $\rho(uv)$ of the addition or deletion of the edge $uv$. For a set of unordered pairs $X\subseteq \binom{V(G)}{2}$, $\rho(X)=\sum_{uv\in X}\rho(uv)$. 
Suppose that a graph $G\rq{}$ is obtained from $G$ by some edge deletions and additions. Then the \emph{editing cost} is $\rho((E(G)\setminus E(G\rq{}))\cup (E(G\rq{})\setminus E(G)))$.
Mathieson and Szeider considered the following problem:
\begin{center}
\begin{boxedminipage}{.99\textwidth}
\textsc{Edge Editing to a Graph of Given Degrees with Costs}\\
\begin{tabular}{ r l }
\textit{~~~~Instance:} & A graph $G$, a non-negative integer $k$, a degree function \\ &
                                   $\delta\colon V(G)\rightarrow\mathbb{N}$ and   a cost function $\rho\colon\binom{V(G)}{2}\rightarrow\mathbb{N}$.\\
\textit{Question:} & Is it possible to obtain a graph $G'$ from $G$ such that \\ &
$d_{G'}(v)=\delta(v)$ for each $v\in V(G')$ by  
edge deletions and \\ & additions with editing cost at most $k$?  \\
\end{tabular}
\end{boxedminipage}
\end{center}
They proved the following theorem.

\begin{theorem}[{\cite[Theorem~5.1]{MathiesonS12}}]\label{thm:edit-cost}
{\sc Edge Editing to a Graph of Given Degrees with Costs} can be solved in polynomial time.
\end{theorem}

We also need some results about graphic sequences for bipartite graphs.
Let $\alpha=(\alpha_1,\ldots,\alpha_p)$ and $\beta=(\beta_1,\ldots,\beta_q)$ be non-increasing sequences of positive integers. We say that the pair $(\alpha,\beta)$ is a \emph{bipartite graphic pair} if there is a bipartite graph $G$ with the bipartition of the vertex set $X=\{x_1,\ldots,x_p\}$, $Y=\{y_1,\ldots,y_q\}$ such that $d_G(x_i)=\alpha_i$ for $i\in \{1,\ldots,p\}$ and $d_G(y_j)=\beta_j$ for $j\in\{1,\ldots,q\}$.  It is said that $G$ \emph{realizes} $(\alpha,\beta)$.

Gale and Ryser~\cite{Ryser63} gave necessary and sufficient conditions for $(\alpha,\beta)$ to be a bipartite graphic pair. It is more convenient to give them in the terms of partitions of integers. 
Recall that a non-increasing sequence of positive integers $\alpha=(\alpha_1,\ldots,\alpha_p)$ is a \emph{partition} of $n$ if $\alpha_1+\ldots+\alpha_p=n$. 
A sequence  
$\alpha=(\alpha_1,\ldots,\alpha_p)$  \emph{dominates}  $\beta=(\beta_1,\ldots,\beta_q)$ if $\alpha_1+\ldots+\alpha_i\geq \beta_1+\ldots+\beta_i$ for all $i\geq 1$; to simplify notations, we
assume that $\alpha_i=0$ ($\beta_i=0$ respectively) if $i>p$ ($i>q$ respectively). We write $\alpha\unrhd\beta$ to denote that $\alpha$ dominates $\beta$. 
Clearly, if $\alpha\unrhd\beta$ and $\beta\unrhd\gamma$, then $\alpha\unrhd\gamma$.
For a partition $\alpha=(\alpha_1,\ldots,\alpha_p)$ of $n$, the partition  $\alpha^*=(\alpha^*_1,\ldots,\alpha_{\alpha_1}^*)$ of $n$, where 
$\alpha_j^*=|\{h|1\leq h\leq p,\alpha_h\geq j\}|$ for $j\in\{1,\ldots,\alpha_1\}$, is called the \emph{conjugate} partition for $\alpha$. Notice that $\alpha^{**}=\alpha$.

\begin{theorem}[Gale and Ryser~\cite{Ryser63}]\label{thm:GR}
A pair of non-increasing sequences of positive integers $(\alpha,\beta)$ is a bipartite graphic pair if and only if $\alpha$ and $\beta$ are partitions of some positive integer $n$ and 
$\alpha^*\unrhd\beta$.
\end{theorem}

By the straightforward reduction to the {\sc Maximum Flow} problem and the well-known fact that it can be solved in polynomial time (see, e.g, \cite{EdmondsK72}), we have the following lemma.

\begin{lemma}\label{lem:graphic-poly}
Let $(\alpha,\beta)$ be a bipartite graphic pair. Then a bipartite graph $G$ that realizes $(\alpha,\beta)$ can be constructed in polynomial time.
\end{lemma}

We also need the following property.

\begin{lemma}\label{lem:graphic}
Let $\alpha,\alpha\rq{},\beta,\beta\rq{}$ be partitions of a positive integer $n$. If $(\alpha,\beta)$, $(\alpha\rq{},\alpha^*)$ and $(\beta^*,\beta\rq{})$ are bipartite graphic pairs, then $(\alpha\rq{},\beta\rq{})$ is a bipartite graphic par.
\end{lemma}

\begin{proof}
Because $(\alpha\rq{},\alpha^*)$ is a bipartite graphic pair, by Theorem~\ref{thm:GR}, $\alpha\rq{}^*\unrhd\alpha^*$. By the same arguments, $\alpha^*\unrhd\beta$ and $\beta^{**}\unrhd\beta\rq{}$. We have that  $\alpha\rq{}^*\unrhd\beta\rq{}$, and by Theorem~\ref{thm:GR}, $(\alpha\rq{},\beta\rq{})$ is a bipartite graphic pair. 
\end{proof}

\subsection{The algorithm for Edge Editing to a Connected Regular Graph}\label{sec:FPT-alg}
Let $(G,d,k)$ be an instance of \textsc{Edge Editing to a Connected Regular Graph}.
We assume that $k\geq 1$, as otherwise the problem is trivial. If $d\leq 3k+1$, then we solve the problem in time $O^*(k^{O(k)})$ by Theorem~\ref{thm:kernel}.
From now it is assumed that $d>3k+1$. 
Let $Z=\{v\in V(G)|d_G(v)\neq d\}$.

First, we check whether $|Z|\leq 2k$ and stop and return a NO-answer otherwise using the observation that each edge deletion (addition respectively) decreases (increases respectively) the degrees of two its end-vertices by one. 
From now  we  assume that $|Z|\leq 2k$.
Denote by $G_1,\ldots,G_p$ the components of $G-Z$.

We say that two components $G_i,G_j$ have the \emph{same type}, if for any $z\in Z$, either $|N_G(z)\cap V(G_i)|=|N_G(z)\cap V(G_j)|\leq k$ or $|N_G(z)\cap V(G_i)|>k$ and $|N_G(z)\cap V(G_j)|>k$.
Denote by $\Theta_1,\ldots,\Theta_t$ the partition of $\{G_1,\ldots,G_p\}$ into classes according to this equivalence relation.
 Observe that the number of distinct types is at most $(k+2)^{2k}$. Notice also that for any solution $(D,A)$, the graph $H(D,A)$ contains vertices of at most $2k$ components $G_1,\ldots,G_p$. 

The general idea of the algorithm is to guess the structure of a possible solution $(D,A)$ (if it exists).  We guess the edges of $D$ and $A$ that join the vertices of $Z$. Then we guess the number and the types of components of $G-Z$ that contain vertices of $H(D,A)$. For them, we guess the number of edges that join these components with each other
and with each vertex of $Z$.  Notice that the edges of $A$ between distinct components of $G-Z$ should form a bipartite graph. Hence, we guess some additional conditions that ensure that such a graph can be constructed.  
Then for each guess, we check in polynomial time whether we have a solution that corresponds to it. The main ingredient here is the fact that we can modify the components of $G-Z$ without destroying their connectivity.  We construct partial solutions for some components of $G-Z$ and then ``glue\rq{}\rq{} them together.

Let $Z=\{z_1,\ldots,z_r\}$. 
We define records $L=(s,\Theta,C,R,D_Z,A_Z)$, where 
\begin{itemize}
\item $0\leq s\leq \min\{2k,p\}$ is an integer, 
\item $\Theta$ is an $s$-tuple $(\tau_1,\ldots,\tau_s)$ of integers and $1\leq \tau_1\leq\ldots\leq\tau_s\leq t$; 
\item $C$ is a $s\times s$ table of bipartite graphic pairs $(\alpha_{j,h},\beta_{j,h})$ with the sum of elements of $\alpha_{j,h}$ denoted $c_{j,h}$
such that $\alpha_{j,h}=\beta_{h,j}$, $0\leq c_{j,h}\leq k$ and $c_{j,j}=0$ for $j,h\in\{1,\ldots,s\}$,  notice that it can happen that $c_{j,h}=0$ and it is assumed that $(\alpha_{j,h},\beta_{j,h})=(\emptyset,\emptyset)$ in this case;
\item $R$ is $r\times s$ integer matrix with the elements $r_{j,h}$ such that $-k\leq r_{j,h}\leq k$ for $j\in\{1,\ldots,r\}$ and $h\in\{1,\ldots,s\}$;
\item $D_Z\subseteq E(G[Z])$; and
\item $A_Z\subseteq \binom{Z}{2}\setminus E(G[Z])$. 
\end{itemize}
Let $(D,A)$ be a solution for $(G,d,k)$.
We say that $(D,A)$ \emph{corresponds} to $L$ if 
\begin{itemize}
\item[i)] the graph $H(D,A)$ contains vertices from exactly $s$ components $G_{i_1},\ldots,G_{i_s}$ of $G-Z$;
\item[ii)] $G_{i_j}\in \Theta_{\tau_j}$ for $j\in\{1,\ldots,s\}$;
\item[iii)] for $j,h\in\{1,\ldots,s\}$, $A$ has exactly $c_{j,h}$ edges between $G_{i_j}$ and $G_{i_h}$ if $j\neq h$;  
\item[iv)] for any $j\in\{1,\ldots,r\}$ and $h\in\{1,\ldots,s\}$, $|\{z_jx\in A|x\in V(G_{i_h})\}|-|\{z_jx\in D|x\in V(G_{i_h})\}|=r_{j,h}$;
\item[v)] $D\cap E(G[Z])=D_Z$;
\item[vi)] $A\cap \binom{Z}{2}=A_Z$.   
\end{itemize} 

It is straightforward to verify that the number of all possible records $L$ is at most $k^{O(k^3)}$. We consider all such records, and for each $L$, we check whether $(G,d,k)$ has a solution that corresponds to $L$. If we find a solution for some $L$, then we stop and return it. Otherwise, if we fail to find any solution, we return a NO-answer. From now we assume that $L$ is given.

For $i\in\{1,\ldots,p\}$,  a given $r$-tuple $Q=(q_1,\ldots,q_r)$ and 
$\ell$-tuple $Q\rq{}=(q_1\rq{},\ldots,q_{\ell}\rq{})$, where 
 $-k\leq q_1,\ldots,q_r\leq k$, $\ell\leq k$ and $1\leq w_1,\ldots,w_{\ell}\leq k$, we consider 
an auxiliary instance $\Pi(i,Q,Q\rq{})$ of {\sc Edge Editing to a Graph of Given Degrees with Costs} defined as follows. We consider the graph $G[Z\cup V(G_i)]$, delete the edges between the vertices of $Z$, and add a set of $\ell$ isolated vertices $W=\{w_1,\ldots,w_{\ell}\}$. Each vertex $w_j$,  we say that it \emph{corresponds to $q_j\rq{}$} for $j\in\{1,\ldots,\ell\}$.
Denote the obtained graph by $F_i$. We set  $\delta(v)=d$ if $v\in V(G_i)$, $\delta(z_j)=d_{F_i}(z_j)+q_j$ for $j\in\{1,\ldots,r\}$, and
$\delta(w_j)=q_j\rq{}$ for $j\in\{1,\ldots,\ell\}$. We set $\rho(uv)=k+1$ if $u,v\in Z\cup W$, and $\rho(uv)=1$ for all other pairs of vertices of $\binom{V(F_i)}{2}$. 
Observe that it can happen that $\delta(z_j)<0$ for some $j\in\{1,\ldots,r\}$. In this case we assume that $\Pi(i,Q,Q\rq{})$ has NO-answer. In all other cases we solve $\Pi(i,Q,Q\rq{})$ and find a solution of minimum editing cost $c(i,Q,Q\rq{})$ using Theorem~\ref{thm:edit-cost}. If we have a NO-instance or $c(i,Q,Q\rq{})>k$, then we set $c(i,Q,Q\rq{})=+\infty$. We need the following property of the solutions. 
The proof 
is based on Lemma~\ref{lem:alt} and uses the fact that for a solution $(D,A)$, $H(D,A)$ can be covered by edge-disjoint $(D,A)$-alternating trails.   

\begin{lemma}\label{lem:connect}
If $c(i,Q,Q\rq{})\leq k$, then any solution for $\Pi(i,Q,Q\rq{})$ of cost at most $k$ has no edges between vertices of $Z\cup W$ and 
there is a solution $(A,D)$  for $\Pi(i,Q,Q\rq{})$ of cost $c(i,Q,Q\rq{})\leq k$ such that if $F\rq{}=F_i-D+A$, then any $u,v\in V(G_i)$ are in the same component of $F\rq{}$. Moreover, such a solution can be found in polynomial time.
\end{lemma}

\begin{proof}
Let $(A,D)$ be a solution of minimum cost. Because $\rho(u,v)=k+1$ for any $u,v\in Z\cup W$, $u\neq v$, $A$ cannot have edges between  vertices of $Z\cup W$. 

Let $F\rq{}=F_i-D+A$.
Consider the graph $H(D,A)$ defined by $(D,A)$. By Lemma~\ref{lem:alt},
$H(D,A)$ can be covered by a family of edge-disjoint $(D,A)$-alternating 
trails $\mathcal{T}$, and each non-closed trail in $\mathcal{T}$ has its end-vertices in $Z\cup W$. 
Recall that $H(D,A)$ has no edges $uv$ for $u,v\in Z\cup W$. Hence, trails have no such edges as well. Because $(D,A)$ is a solution of minimum cost, $\mathcal{T}$ has no $(D,A)$-alternating closed trails, as otherwise the edges of such a trail could be excluded from $D$ and $A$ respectively without changing the degrees. 

We say that a trail $P\in \mathcal{T}$ is \emph{simple} if it begins and ends by edges from $A$ and has the unique edge from $D$, and this edge is in $G_i$. 

Let $xy\in D$ be an edge of a non-simple trail $P$ 
such that $x,y\in V(G_i)$. Since $P$ is not a simple trail, $xy$ is in a subtrail $u,ux,x,xy,y,yv,v$, where $ux,yv\in A$ and $u\in V(G_i)$ or $v\in V(G_i)$. By symmetry, assume that $v\in V(G_i)$. We have that $uv\in E(F\rq{})$, because otherwise $u,ux,x,xy,y,yv,v$ could be replaced by $u,uv,v$ in $P$, and it would give a better solution, as $\rho(uv)=1$. Moreover, observe that if we modify the solution by replacing some simple trail by another simple trail with the same end-vertices, this modification cannot remove $uv$, as it would again imply that we can improve the solution. Then we have that the end-vertices of the deleted edge $xy$ are connected by a path in $F\rq{}$.

Let $G_i\rq{}$ be the graph obtained from $G_i$ by the deletion of the edges of $D$. Denote by $F_1,\ldots,F_h$ the components of $G_i'$. If for any simple trail $P=u,ux,x,xy,y,yv,v$, $x$ and $y$ are in the same component, then the claim is proved, because $x$ and $y$ are joined by a path in $G_i'$ and, therefore, in $F'$. Assume that for some simple trail  $P=u,ux,x,xy,y,yv,v$, $x$ and $y$ are in different components. Without loss of generality we assume that $x\in V(F_1)$ and $y\in V(F_2)$.  We have the following case.

\medskip
\noindent{\bf Case 1.} The vertex $u$ is joined by an edge with a vertex in $F_2$ or  
$v$ is joined by an edge with a vertex in $F_1$ in the graph $F'$. Then $x$ and $y$ are joined by a path in $F'$.

\medskip
\noindent{\bf Case 2.} The vertex $u$ is not adjacent to the vertices of $F_2$ and $v$ is not adjacent to the vertices of $F_1$, but there is $x'\in V(F_1')$ such that $x'\neq x$ and $ux'\in E(F')$ or  there is $y'\in V(F_y')$ such that $y'\neq y$ and $uy'\in E(F')$. We replace $P$ by $P'=u,uy,y,yx,x,xv,v$ and modify $A$ by replacing $ux,vy$ by $uy,vx$. Clearly, this modification gives us another solution with the same cost, and now $x$ and $y$ are joined by a path in $F'$ that is modified respectively.

\medskip
\noindent{\bf Case 3.} The vertex $u$ is not adjacent to the vertices of $F_2$, $v$ is not adjacent to the vertices of $F_1$, and $ux,vy$ are the unique edges that join $u,v$ with $F_1,F_2$ respectively in $F'$.
Since $(D,A)$ is a solution, there are at most $k/2$ edges of $D$ in $G_i$. 
Also all vertices of $G_i$ have the same degree $d$ in $G$.
Therefore,
$\sum_{v\in V(F_1)}d_{F_1}(v)\geq (d-2k)|V(F_1)|-k\geq (k+1)|V(F_1)|-k>(k+1)(|V(F_1)|-1)\geq 2(|V(F_1)|-1) $, and $F_1$ has a cycle. Hence, $F_1$ has an edge $x'y'$ such that $F_1'-x'y'$ is connected.  
We replace $P$ by $P'=u,ux',x',x'y',y',y'v,v$ and modify $D$ by replacing $xy$ by $x'y'$ and 
 $A$ by replacing $ux,vy$ by $ux',vy'$. Clearly, this modification gives us another solution with the same cost, and now $x'$ and $y'$ are joined by a path in $F'$ that is modified respectively.

By applying the same modification for all simple trails, we obtain the solution with the property that for any $xy\in D\cap E(G_i)$, $x$ and $y$ are joined by a path in $F'$. 

To complete the proof, it remains to observe that by Theorem~\ref{thm:edit-cost}, an initial solution of minimum cost can be found in polynomial time. Then $\mathcal{T}$ can be constructed in polynomial time by Lemma~\ref{lem:alt}. Finally, it is straightforward to see that simple paths can be modified in polynomial time. 
\end{proof}

Now we are ready to describe the algorithm that for a record $L=(s,\Theta,C,R,D_Z,A_Z)$, checks whether  $(G,d,k)$ has a solution that corresponds to $L$.

First, we check whether the modification of $G$ with respect to $L$ would satisfy the degree restrictions for $Z$, as otherwise we have no solution. Also the number of edges between $G_1,\ldots,G_p$ should be at most $k$.

\medskip
\noindent
{\bf Step 1.} Let $\hat{G}$ be the graph obtained from $G$ by the deletion of the edges of $D_Z$ and the addition the edges of $A_Z$.  If for any $j\in\{1,\ldots,r\}$, $d_{\hat{G}}(z_j)+\sum_{h=1}^sr_{j,h}\neq d$, then stop and return a NO-answer. 

\medskip
\noindent
{\bf Step 2.} If $\sum_{1\leq j<h\leq s} c_{j,h}>k$, then stop and return a NO-answer. 

\medskip
From now we assume that the degree restrictions for $Z$ are fulfilled and the number of added edges between the components of $G-Z$ should be at most $k$.

\medskip
\noindent
{\bf Step 3.} Construct an auxiliary weighted bipartite graph $F$, where $X=\{x_1,\ldots,x_s\}$ and $Y=\{y_1,\ldots,y_p\}$ is the bipartition of the vertex set. For $i\in \{1,\ldots,s\}$ and $j\in\{1,\ldots,p\}$, we construct an edge $x_iy_j$ if $G_j\in \Theta_{\tau_i}$.  To define the weight $w(x_iy_j)$, we consider $\Pi(j,Q_j,Q_j\rq{})$ where
$Q_j=(r_{1,i},\ldots,r_{r,i})$ and $Q_j\rq{}$ is the sequence obtained by the concatenation of non-empty sequences $\alpha_{j,1}^*,\ldots,\alpha_{j,s}^*$.  
Denote by $W_{j,h}$ the set of vertices of the graph in $\Pi(j,Q_j,Q_j\rq{})$ corresponding to the elements of $\alpha_{j,h}^*$.
Notice that by Step 2, $Q_j\rq{}$ has at most $k$ elements. We set $w(x_iy_j)=c(j,Q_j,Q_j\rq{})$. Observe that some edges can have infinite weights.

\medskip
\noindent
{\bf Step 4.} Find a perfect matching $M$ in $F$ with respect to $X$ of minimum weight. If $F$ has no perfect matching of finite weight, then the algorithm stops and returns a NO-answer.  
Assume that $M=\{x_1y_{j_1},\ldots x_sy_{j_s}\}$ is a perfect matching of minimum weight $\mu<+\infty$. 
If $\mu-\sum_{1\leq j<h\leq s}c_{j,h}+|D_Z|+|A_Z|>k$, then we stop and return a NO-answer. 
\medskip

Now we assume that $M$ has weight at most $k$.

\medskip
\noindent
{\bf Step 5.}
Consider the solutions $(D_{i},A_{i})$ of cost $c(j,Q_{j_i},Q_{j_i}\rq{})$ for  $\Pi(j_i,Q_{j_i},Q_{j_i}\rq{})$
for $i\in\{1,\ldots,s\}$. 

Set $D=D_Z\cup( \cup_{i=1}^sD_i)$.

Construct a set $A$ as follows. 
For $i\in \{1,\ldots,s\}$, denote by $A_i\rq{}$ the set of edges of $A_i$ with the both end-vertices in $V(G_{j_i})\cup Z$, and let $A_{i,h}$ be the subset of edges that join $G_{j_i}$ 
with $W_{j_i,j_h}$ for $h\in\{1,\ldots,s\}$, $h\neq j$.
Initially we include in $A$ the set $\cup_{i=1}^sA_i\rq{}$. For each pair of indices $i,h\in\{1,\ldots,s\}$, such that $i<h$ and $c_{j_i,j_h}>0$, 
consider graphs induced by $A_{i,h}$ and $A_{h,i}$ respectively, 
and denote by $\alpha_{i,h}\rq{}$ and $\alpha_{h,i}\rq{}$ respectively the degree sequences of these graphs for the vertices in $G_{j_i}$ and $G_{j_h}$ respectively. By the construction of the problems    
$\Pi(j,Q_j,Q_j\rq{})$, $(\alpha_{i,h}\rq{},\alpha_{i,h}^*)$ and  $(\alpha_{h,i}\rq{},\alpha_{h,i}^*)$ are bipartite graphic pairs. Recall that $\beta_{i,h}=\alpha_{h,i}$ and $(\alpha_{i,h},\beta_{i,h})$ is a bipartite graphic pair. By Lemma~\ref{lem:graphic}, $(\alpha_{i,h}\rq{},\alpha_{h,i}\rq{})$ is a bipartite graphic pair. Construct a bipartite graph that realizes  $(\alpha_{i,h}\rq{},\alpha_{h,i}\rq{})$
using Lemma~\ref{lem:graphic-poly} and denote its set of edges by $A_{i,h}\rq{}$. We use the vertices of $G_{j_i}$ and $G_{j_h}$ incident with the vertices of $A_{i,h}$ and $A_{h,i}$ as the sets of bipartition and construct our bipartite graph in such a way that for each vertex, the number of edges of  $A_{i,h}\rq{}$ incident to it is the same as the number of edges of  $A_{i,h}$ or $A_{h,i}$ respectively incident to this vertex. Then we include the edges of $A_{i,h}\rq{}$ in $A$.

\medskip
\noindent
{\bf Step 6.} For each $i\in\{1,\ldots,s\}$, do the following. Consider the set of vertices $W_i=\{w_1,\ldots,w_{\ell}\}$ of $\cup_{h=1}^sV(G_{j_h})\setminus V(G_{j_i})$ incident to the edges of $A$ that join $G_{j_i}$ with these vertices and let $q_h\rq{}$ be the number of edges of $A$ that join $G_{j_i}$ with $w_{h}$ for $h\in\{1,\ldots,\ell\}$. Consider  
 $\Pi(j_i,Q_{i},Q_i\rq{})$ where
$Q_i=(r_{1,j_i},\ldots,r_{r,j_i})$ and $Q_i\rq{}=(q_1\rq{},\ldots,q_{\ell}\rq{})$. Using Lemma~\ref{lem:connect}, find a solution $(D_i,A_i)$ for $\Pi(j_i,Q_i,Q_i\rq{})$ of minimum cost. Modify $(D,A)$ by replacing the edges of $D$ and $A$ incident to the vertices of $G_{j_i}$ by the edges of $D_i$ and $A_i$ respectively identifying the set $W_i$ and the set of vertices $W$ in  $\Pi(j_i,Q_{i},Q_i\rq{})$. 

\medskip
\noindent
{\bf Step 7.} Let $G\rq{}=G-D+A$. If $G$ is connected, then return $(D,A)$. Otherwise return a NO-answer. 

\medskip
Suppose that the algorithm produces the sets $(D,A)$ in Step 7. By the description of the algorithm, for any vertex $v\in V(G\rq{})$, $d_{G\rq{}}(v)=d$. Notice that for the sets $D$ and $A$ obtained in Step 5, $|D|+|A|=\mu-\sum_{1\leq j<h\leq s}c_{j,h}+|D_Z|+|A_Z|\leq k$, and in Step 6 $|D|+|A|$ could be only decreased. Since $G\rq{}$ is connected, we conclude that $(D,A)$ is a solution for $(G,d,k)$. It is straightforward to see that $(D,A)$ corresponds to $L$.

Assume now that $(G,d,k)$ has some solution $(D,A)$. The graph $H(D,A)$ contains vertices of some components $G_{i_1},\ldots,G_{i_s}$ of $G-Z$ for $s\leq 2k$. Assume that $G_{i_h}\in \Theta_{\tau_h}$ for $h\in\{1,\ldots,s\}$, and let $\Theta=(\tau_1,\ldots,\tau_s)$.  
Let $D_Z=D\cap E(G[Z])=D_Z$ and 
$A_Z=A\cap \binom{Z}{2}$.
For $j\in\{1,\ldots,s\}$, let $A_j$ be the subset of edges of $A\setminus A_Z$ with the both end-vertices in $V(G_{i_j})\cup Z$ and $D_j=D\cap E( G[V(G_{i_j})\cup Z] )\setminus D_Z$. Consider each pair of indices $j,h\in\{1,\ldots,s\}$, $j\neq h$.  Denote by $A_{j,h}$ the set of edges of $A$ that join $G_{i_j}$ and $G_{i_h}$. The set $A_{j,h}$ induces a bipartite graph. Let $\alpha_{j,h}$ and $\beta_{j,h}$ be the graphic sequences of the vertices of this graph in $G_{i_j}$ and $G_{i_h}$ respectively
(if $A_{j,h}=\emptyset$, then  $\alpha_{j,h}=\beta_{j,h}\emptyset$). Also let $c_{j,h}=|A_{j,h}|$.    Denote by $C$ the table of pairs $(\alpha_{j,h},\beta_{j,h})$.
For any $j\in\{1,\ldots,r\}$ and $h\in\{1,\ldots,s\}$, let $r_{j,h}=|\{z_jx\in A|x\in V(G_{i_h})\}|-|\{z_jx\in D|x\in V(G_{i_h})\}|$. Denote by $R$ the matrix with these elements.  We consider the record $L=(s,\Theta,C,R,D_Z,A_Z)$ and analyze our algorithm for it. 

It is straightforward to see that the algorithm does not stop in Steps 1 and 2. 

Consider now the auxiliary graph $F$ constructed in Step 3. Clearly, $x_1y_{i_1},\ldots,x_sy_{i_s}\in E(F)$. We claim that for each $j\in \{1,\ldots,s\}$, 
$w(x_jy_{i_j})\leq |D_j|+|A_j|+|A_{j,1}|+\ldots+|A_{j,j-1}|+|A_{j,j+1}|+\ldots+|A_{j,s}|\leq k$. To see this, notice that for any non-increasing sequence of positive integers $\alpha$, $(\alpha,\alpha^*)$ is a bipartite graphic pair by Theorem~\ref{thm:GR}. It implies that a feasible solution $(D\rq{},A\rq{})$ for 
$\Pi(i_j,Q_{i_j},Q_{i_j}\rq{})$ can be constructed as follows. Let $D\rq{}=D_{j}$. To construct $A\rq{}$, we include first in this set the edges of $A_{j}$. 
For each $h\in\{1,\ldots,s\}$ such that $h\neq j$ and $c_{j,h}>0$, 
the vertices in $W_{j,h}$ are joined with the vertices of $G_{i_j}$ incident to the edges of $A_{j,h}$ by a set of edges $A_{j,h}\rq{}$ in such a way that 
the graph induced by $A_{j,h}\rq{}$ realizes the pair $(\alpha_{j,h},\alpha_{j,h}^*)$ and the number of edges of $A_{j,h}\rq{}$ incident to each vertex of $G_{i_j}$ is the same as the number of edges of $A_{j,h}$ incident to it. By the definition of $\Pi(i_j,Q_{i_j},Q_{i_j}\rq{})$, we have a feasible solution. 
It follows that  $\{x_1y_{i_1},\ldots,x_sy_{i_s}\}$ is a perfect matching in $F$ of weight at most $|D|+|A|-|D_Z|-|A_Z|+\sum_{1\leq j<h\leq s}c_{j,h}$. 
Therefore, $F$ has a perfect matching $M=\{x_1y_{j_1},\ldots x_sy_{j_s}\}$ of minimum weight $\mu$ such that $\mu-\sum_{1\leq j<h\leq s}c_{j,h}+|D_Z|+|A_Z|\leq k$. In particular, it means that we do not stop in Step 4.

Denote by $D\rq{},A\rq{}$ the sets constructed in Step 5 (and denoted $D$ and $A$ respectively in the description). By the construction, for the graph $G\rq{}=G-D\rq{}+A\rq{}$, $d_{G\rq{}}(v)=d$ for any $v\in V(G\rq{})$. Moreover, by the construction of Step 6, the modifications of $(D\rq{},A\rq{})$ maintain this property. Hence, to show that we obtain a solution in Step 7, it remains to show that Step 6 ensures that we get a connected graph $G\rq{}$ if we delete the edges of the modified set $D\rq{}$ and add the edges of the modified set $A\rq{}$.

For each $i\in\{1,\ldots,s\}$, we inductively prove the following. Let $D^i,A^i$ be the sets  of edges constructed in Step 6 after executing the first $i$ modification of the sets obtained in Step 5.  Let also $G^i=G-D^i+A^i$.
Then for any $h\in\{1,\ldots,i\}$ and any $u,v\in V(G_{j_i})$, $G^i$ has a $(u,v)$-path $P$ such that $P$ has no edges in $G^i- V(G_{j_1})\cup\ldots\cup V(G_{j_i})$ (but $P$ can have vertices
in this set).  
For $i=1$, the claim immediately follows from Lemma~\ref{lem:connect}. Assume now that $i>1$. 
If $u,v\in V(G_{j_i})$, then we again apply Lemma~\ref{lem:connect}. Let $u,v\in V(G_{j_h})$ for $h<i$. By the inductive hypothesis, $u$ and $v$ could be connected by some path $P$ without edges in  $G^{i-1}- V(G_{j_1})\cup\ldots\cup V(G_{j_{i-1}})$. Suppose that this path is destroyed by the further modifications. It can happen only if $P$ has subpaths $xyz$ where $x,z\in 
V(G_{j_1})\cup\ldots\cup V(G_{j_{i-1}})$ and $y\in V(G_{j_i})$ and $xy$ or $yz$ is not in $G^i$. But then there are $y_1,y_2\in V(G_{j_i})$ such that $xy_1,zy_2\in A^i$ and we have the required connectivity between $y_1$ and $y_2$. Then the claim follows.

By this claim, for any $u,v\in V(G_{j_h})$, $u$ and $v$ can be connected by a path in $G\rq{}$ for $h\in\{1,\ldots,s\}$. 
 Recall that $(D,A)$ is a solution for $(G,d,k)$, and $H(D,A)$ has vertices from $G_{i_1},\ldots,G_{i_s}$. For each $h\in \{1,\ldots,s\}$, $G_{i_h}$ and $G_{j_h}$ have the same type. 
Also for each $i\in\{1,\ldots,r\}$, the number of edges that join $z_i$ and $G_{i_h}$ with respect  to the solution $(D,A)$ and the number  of edges that join $z_i$ and $G_{j_h}$ with respect to $(D\rq{},A\rq{})$ is increased (or decreased if the number is negative) by $r_{i,h}$. Because $G_{i_h}$ and $G_{j_h}$ have the same type, there are no edges that join $z_i$ and $G_{i_h}$ with respect  to $(D,A)$ if and only if are no edges that join $z_i$ and $G_{j_h}$ with respect  to $(D\rq{},A\rq{})$ in the graphs obtained by editing. Furthermore, for any
$g,h\in \{1,\ldots,s\}$, $C_{i_g}$ and $C_{i_h}$ in $(D,A)$ and $C_{j_g}$ and $C_{j_h}$ in $(D\rq{},A\rq{})$ are connected by the same number of edges $c_{g,h}$. It implies that if the graph obtained from $G$ by editing with respect to $(D,A)$ is connected, then $G\rq{}$ obtained by editing with respect to $(D\rq{},A\rq{})$ is also connected, and we have that $(D\rq{},A\rq{})$ produced by the algorithm is a solution for $(G,d,k)$. Also we can observe that this solution corresponds to $L$. 

Now we argue that this algorithm is polynomial. Clearly, Steps 1 and 2 can be performed in polynomial time.  The construction of the graph $F$ in Step 3 can be done polynomially, and the weight assignment demands polynomial time, because the construction of the auxiliary problem can be done in polynomial time, and we can solve this problem in polynomial time by Lemma~\ref{lem:connect}.  We can find a perfect matching of minimum weight in $F$ by standard algorithms (see, e.g, \cite{EdmondsK72}). The construction of $(D,A)$ in Step 5 can be done in polynomial time by Lemma~\ref{lem:graphic-poly}. Step 6 is polynomial by Lemma~\ref{lem:connect}. As connectedness in Step 7 can be easily checked, this step is also polynomial. 

To complete the proof of Theorem~\ref{thm:regular}, it remains to observe that to solve an instance of  {\sc Edge Editing to a Graph of Given Degrees with Costs}, we generate at most $k^{O(k^3)}$ records, and run a polynomial algorithm for these records. It follows that {\sc Edge Editing to a Graph of Given Degrees with Costs} can be solved in time $O^*(k^{O(k^3)})$.

\section{Conclusions}
We proved that \textsc{Editing to a Graph of Given Degrees}  has a polynomial kernel of size $O(kd^3(k+d)^2)$.  It is natural to ask whether the size can be improved. 
Also, is the problem \classFPT\ when parameterized by $k$ only? We proved that it holds for the special case $\delta(v)=d$, i.e., for \textsc{Edge Editing to a Connected Regular Graph}.
Another open question is whether \textsc{Editing to a Graph of Given Degrees} (or \textsc{Edge Editing to a Connected Regular Graph})  has a polynomial kernel with the size that depends on $k$ only.

\end{document}

%% file: Fig1.pdf_t
\begin{picture}(0,0)%
\includegraphics{Fig1.pdf}%
\end{picture}%
\setlength{\unitlength}{3947sp}%
\begingroup\makeatletter\ifx\SetFigFont\undefined%
\gdef\SetFigFont#1#2#3#4#5{%
  \reset@font\fontsize{#1}{#2pt}%
  \fontfamily{#3}\fontseries{#4}\fontshape{#5}%
  \selectfont}%
\fi\endgroup%
\begin{picture}(4827,1920)(1264,-2323)
\put(1576,-1636){\makebox(0,0)[lb]{\smash{{\SetFigFont{12}{14.4}{\rmdefault}{\mddefault}{\updefault}{\color[rgb]{0,0,0}$F_{\ell}$}%
}}}}
\put(1351,-2236){\makebox(0,0)[lb]{\smash{{\SetFigFont{12}{14.4}{\rmdefault}{\mddefault}{\updefault}{\color[rgb]{0,0,0}$Z$}%
}}}}
\put(4351,-1711){\makebox(0,0)[lb]{\smash{{\SetFigFont{12}{14.4}{\rmdefault}{\mddefault}{\updefault}{\color[rgb]{0,0,0}$G[N_G^2[Z]\cap V(G_i)]$}%
}}}}
\put(1951,-1636){\makebox(0,0)[lb]{\smash{{\SetFigFont{12}{14.4}{\rmdefault}{\mddefault}{\updefault}{\color[rgb]{0,0,0}$u_{\ell}$}%
}}}}
\put(1951,-886){\makebox(0,0)[lb]{\smash{{\SetFigFont{12}{14.4}{\rmdefault}{\mddefault}{\updefault}{\color[rgb]{0,0,0}$v_{\ell}$}%
}}}}
\put(3751,-1636){\makebox(0,0)[lb]{\smash{{\SetFigFont{12}{14.4}{\rmdefault}{\mddefault}{\updefault}{\color[rgb]{0,0,0}$u_1$}%
}}}}
\put(4051,-886){\makebox(0,0)[lb]{\smash{{\SetFigFont{12}{14.4}{\rmdefault}{\mddefault}{\updefault}{\color[rgb]{0,0,0}$v_0$}%
}}}}
\put(3751,-886){\makebox(0,0)[lb]{\smash{{\SetFigFont{12}{14.4}{\rmdefault}{\mddefault}{\updefault}{\color[rgb]{0,0,0}$v_1$}%
}}}}
\put(4726,-1186){\makebox(0,0)[lb]{\smash{{\SetFigFont{12}{14.4}{\rmdefault}{\mddefault}{\updefault}{\color[rgb]{0,0,0}$x_1$}%
}}}}
\put(6076,-1186){\makebox(0,0)[lb]{\smash{{\SetFigFont{12}{14.4}{\rmdefault}{\mddefault}{\updefault}{\color[rgb]{0,0,0}$x_s$}%
}}}}
\put(4726,-586){\makebox(0,0)[lb]{\smash{{\SetFigFont{12}{14.4}{\rmdefault}{\mddefault}{\updefault}{\color[rgb]{0,0,0}$y_1$}%
}}}}
\put(6076,-586){\makebox(0,0)[lb]{\smash{{\SetFigFont{12}{14.4}{\rmdefault}{\mddefault}{\updefault}{\color[rgb]{0,0,0}$y_s$}%
}}}}
\put(3376,-1636){\makebox(0,0)[lb]{\smash{{\SetFigFont{12}{14.4}{\rmdefault}{\mddefault}{\updefault}{\color[rgb]{0,0,0}$F_1$}%
}}}}
\end{picture}%

%% file: Fig9.pdf_t
\begin{picture}(0,0)%
\includegraphics{Fig9.pdf}%
\end{picture}%
\setlength{\unitlength}{3947sp}%
\begingroup\makeatletter\ifx\SetFigFont\undefined%
\gdef\SetFigFont#1#2#3#4#5{%
  \reset@font\fontsize{#1}{#2pt}%
  \fontfamily{#3}\fontseries{#4}\fontshape{#5}%
  \selectfont}%
\fi\endgroup%
\begin{picture}(8120,2106)(143,-1784)
\put(6526,-1711){\makebox(0,0)[lb]{\smash{{\SetFigFont{12}{14.4}{\rmdefault}{\mddefault}{\updefault}{\color[rgb]{0,0,0}b)}%
}}}}
\put(3526,-1336){\makebox(0,0)[lb]{\smash{{\SetFigFont{12}{14.4}{\rmdefault}{\mddefault}{\updefault}{\color[rgb]{0,0,0}$v$}%
}}}}
\put(7801,-1336){\makebox(0,0)[lb]{\smash{{\SetFigFont{12}{14.4}{\rmdefault}{\mddefault}{\updefault}{\color[rgb]{0,0,0}$v$}%
}}}}
\put(1576,-1336){\makebox(0,0)[lb]{\smash{{\SetFigFont{12}{14.4}{\rmdefault}{\mddefault}{\updefault}{\color[rgb]{0,0,0}$v$}%
}}}}
\put(5401,-1336){\makebox(0,0)[lb]{\smash{{\SetFigFont{12}{14.4}{\rmdefault}{\mddefault}{\updefault}{\color[rgb]{0,0,0}$v$}%
}}}}
\put(2551,-1711){\makebox(0,0)[lb]{\smash{{\SetFigFont{12}{14.4}{\rmdefault}{\mddefault}{\updefault}{\color[rgb]{0,0,0}a)}%
}}}}
\end{picture}%

%% file: Fig8.pdf_t
\begin{picture}(0,0)%
\includegraphics{Fig8.pdf}%
\end{picture}%
\setlength{\unitlength}{3947sp}%
\begingroup\makeatletter\ifx\SetFigFont\undefined%
\gdef\SetFigFont#1#2#3#4#5{%
  \reset@font\fontsize{#1}{#2pt}%
  \fontfamily{#3}\fontseries{#4}\fontshape{#5}%
  \selectfont}%
\fi\endgroup%
\begin{picture}(4220,1952)(4118,-1630)
\put(7351,-1561){\makebox(0,0)[lb]{\smash{{\SetFigFont{12}{14.4}{\rmdefault}{\mddefault}{\updefault}{\color[rgb]{0,0,0}$u$}%
}}}}
\put(4951,-1561){\makebox(0,0)[lb]{\smash{{\SetFigFont{12}{14.4}{\rmdefault}{\mddefault}{\updefault}{\color[rgb]{0,0,0}$u$}%
}}}}
\put(5551,-1561){\makebox(0,0)[lb]{\smash{{\SetFigFont{12}{14.4}{\rmdefault}{\mddefault}{\updefault}{\color[rgb]{0,0,0}$v$}%
}}}}
\put(7951,-1561){\makebox(0,0)[lb]{\smash{{\SetFigFont{12}{14.4}{\rmdefault}{\mddefault}{\updefault}{\color[rgb]{0,0,0}$v$}%
}}}}
\end{picture}%

%% file: Editing-arxiv.bbl
\begin{thebibliography}{10}

\bibitem{AlonSS05}
{\sc N.~Alon, A.~Shapira, and B.~Sudakov}, {\em Additive approximation for
  edge-deletion problems}, in FOCS, IEEE Computer Society, 2005, pp.~419--428.

\bibitem{BurzynBD06}
{\sc P.~Burzyn, F.~Bonomo, and G.~Dur{\'a}n}, {\em Np-completeness results for
  edge modification problems}, Discrete Applied Mathematics, 154 (2006),
  pp.~1824--1844.

\bibitem{Cai96}
{\sc L.~Cai}, {\em Fixed-parameter tractability of graph modification problems
  for hereditary properties}, Inf. Process. Lett., 58 (1996), pp.~171--176.

\bibitem{DowneyF13}
{\sc R.~G. Downey and M.~R. Fellows}, {\em Fundamentals of Parameterized
  Complexity}, Texts in Computer Science, Springer, 2013.

\bibitem{EdmondsK72}
{\sc J.~Edmonds and R.~M. Karp}, {\em Theoretical improvements in algorithmic
  efficiency for network flow problems}, J. ACM, 19 (1972), pp.~248--264.

\bibitem{FlumG06}
{\sc J.~Flum and M.~Grohe}, {\em Parameterized complexity theory}, Texts in
  Theoretical Computer Science. An EATCS Series, Springer-Verlag, Berlin, 2006.

\bibitem{GareyJ79}
{\sc M.~R. Garey and D.~S. Johnson}, {\em Computers and Intractability: A Guide
  to the Theory of NP-Completeness}, W. H. Freeman, 1979.

\bibitem{KhotR02}
{\sc S.~Khot and V.~Raman}, {\em Parameterized complexity of finding subgraphs
  with hereditary properties}, Theor. Comput. Sci., 289 (2002), pp.~997--1008.

\bibitem{LewisY80}
{\sc J.~M. Lewis and M.~Yannakakis}, {\em The node-deletion problem for
  hereditary properties is np-complete}, J. Comput. Syst. Sci., 20 (1980),
  pp.~219--230.

\bibitem{MathiesonS12}
{\sc L.~Mathieson and S.~Szeider}, {\em Editing graphs to satisfy degree
  constraints: A parameterized approach}, J. Comput. Syst. Sci., 78 (2012),
  pp.~179--191.

\bibitem{MoserT09}
{\sc H.~Moser and D.~M. Thilikos}, {\em Parameterized complexity of finding
  regular induced subgraphs}, J. Discrete Algorithms, 7 (2009), pp.~181--190.

\bibitem{NatanzonSS01}
{\sc A.~Natanzon, R.~Shamir, and R.~Sharan}, {\em Complexity classification of
  some edge modification problems}, Discrete Applied Mathematics, 113 (2001),
  pp.~109--128.

\bibitem{Niedermeierbook06}
{\sc R.~Niedermeier}, {\em Invitation to fixed-parameter algorithms}, vol.~31
  of Oxford Lecture Series in Mathematics and its Applications, Oxford
  University Press, Oxford, 2006.

\bibitem{Ryser63}
{\sc H.~J. Ryser}, {\em Combinatorial mathematics}, The Carus Mathematical
  Monographs, No. 14, Published by The Mathematical Association of America,
  1963.

\bibitem{Yannakakis78}
{\sc M.~Yannakakis}, {\em Node- and edge-deletion np-complete problems}, in
  STOC, R.~J. Lipton, W.~A. Burkhard, W.~J. Savitch, E.~P. Friedman, and A.~V.
  Aho, eds., ACM, 1978, pp.~253--264.

\end{thebibliography}
